\documentclass[12pt]{article}
\usepackage{secdot}
\usepackage{natbib}
\usepackage{ifthen}
 \bibpunct[, ]{(}{)}{,}{a}{}{,}%
 \def\bibfont{\small}%
\usepackage{setspace}

\usepackage{amsthm}
\usepackage[geometry]{ifsym}

\usepackage[boxed,lined,noend,algoruled]{algorithm2e}
\usepackage{amsmath,amssymb}
\usepackage{tikz}
\usepackage{tkz-euclide}
\usepackage{todonotes,color}
\usepackage{url}
\usepackage{subfig}
\usepackage{pgfplots}
\usepackage{xcolor}
\usepackage{xspace}
\usepackage{framed}
\usepackage[T1]{fontenc}
\DeclareSymbolFont{largesymbolsA}{U}{txexa}{m}{n}
\DeclareMathSymbol{\varprod}{\mathop}{largesymbolsA}{16}

\DeclareMathAlphabet\mathbfcal{OMS}{cmsy}{b}{n}
\def\mathbi#1{\textbf{\em #1}}

\newcommand{\set}[2]{\left\{#1 \; \left|\;\; #2 \right.\right\}}

\newcommand{\tchak}{\vdash}

\newcommand{\exact}{\textbf{exact}}
\newcommand{\hcenter}{\textbf{center}}
\newcommand{\cons}{\textbf{cons}}
\newcommand{\worst}{\textbf{worst}}
\newcommand{\avg}{\textbf{avg}}
\newcommand{\compact}{\textbf{compact}}

\newcommand{\Metric}{\mathcal{M}}
\newcommand{\w}{{w}}
\newcommand{\C}{\mathcal{C}}

\newcommand{\GG}{\mathbb{G}}
\newcommand{\VV}{\mathbb{V}}
\newcommand{\EE}{\mathbb{E}}

\newcommand{\R}{\mathbb{R}}
\newcommand{\Edir}{\Vec{\EE}}

\newcommand{\U}{\mathcal{U}}

\newcommand{\W}{\mathcal{W}^p}
\newcommand{\WW}{\mathcal{W}}
\newcommand{\UU}{\mathcal{U}}
\newcommand{\tUU}{\widetilde{\mathcal{U}}}

\newcommand{\tomega}{\tilde{\omega}}
\newcommand{\tx}{\tilde{x}}
\newcommand{\tmu}{\tilde{\mu}}

\newcommand{\tu}{\tilde{u}}

\newcommand{\NP}{\mathcal{NP}}
\newcommand{\X}{\mathcal{X}}
\newcommand{\grandO}{\mathcal{O}}
\renewcommand{\u}{u}
\renewcommand{\v}{v}
\newcommand{\bc}{\beta}
\newcommand{\dd}{\hat{d}}

\newcommand{\up}{{u}^+}
\newcommand{\un}{{u}^-}

\newcommand{\nU}[1]{{\sigma_{#1}}}
\newcommand{\NU}{\sigma}

\newcommand{\diaminst}{\Delta}

\newcommand{\roots}{R}
\renewcommand{\root}{r}

\newcommand{\dmax}{d^{max}}

\newcommand{\cost}{c}

\newcommand{\card}[1]{\left\lvert #1\right\rvert}

\newcommand{\distany}[2]{d(#1,#2)}
\newcommand{\dist}[2]{\distany{#1}{#2}}

\DeclareMathOperator*{\conv}{conv}

\renewcommand{\mp}[1]{{\color{brown}mp : #1}}

\newcommand{\SP}{\textsc{sp}\xspace}
\newcommand{\MST}{\textsc{mst}\xspace}

\newcommand{\PBFAMILY}{\ensuremath{\mathcal{S}}\xspace}
\newcommand{\EVALC}{\ensuremath{\textsc{adversarial}}\xspace}
\newcommand{\MAXCUT}{\ensuremath{\textsc{max-cut}}\xspace}
\newcommand{\LISTCOL}{\ensuremath{\textsc{list-col}}\xspace}

\newcommand{\ROBUST}[1]{\ensuremath{\textsc{LocRob}\mbox{\footnotesize-#1}}\xspace}
\newcommand{\CONSERVATIVE}[1]{\ensuremath{\textsc{cons}\mbox{\footnotesize-#1}}\xspace}

\newcommand{\OPT}{\mbox{\textsc{opt}}}

\newcommand{\NPH}{$\cal{NP}$-hard\xspace}
\newcommand{\PEQNP}{$\cal{P}=\cal{NP}$\xspace}
\newcommand{\PTAS}{$\cal{PTAS}$\xspace}
\newcommand{\WONEH}{${\cal W}[1]$-hard\xspace}
\newcommand{\WONEHness}{${\cal W}[1]$-hardness\xspace}
\newcommand{\FPT}{$\cal{FPT}$\xspace}

\DeclareMathOperator*{\argmin}{arg\,min}

\newtheorem{theorem}{Theorem}

\newtheorem{example}{Example}

\newtheorem{lemma}{Lemma}

\newtheorem{proposition}{Proposition}
\newtheorem{remark}{Remark}

\newcommand{\blue}[1]{{\color{black}#1}}

\usepackage{hyperref}

\begin{document}
%%%%%%%%%%%%%%%%

\begin{center}

{\LARGE Optimization problems in graphs with locational uncertainty}\\[12pt]

\footnotesize

\mbox{\large Marin Bougeret}\\ LIRMM, University of Montpellier, CNRS, France, \mbox{marin.bougeret@lirmm.fr}\\[6pt]

\mbox{\large J\'er\'emy Omer}\\ IRMAR, INSA de Rennes, Rennes, France, \mbox{jeremy.omer@insa-rennes.fr}\\[6pt]

\mbox{\large Michael Poss}\\ LIRMM, University of Montpellier, CNRS, France, \mbox{michael.poss@lirmm.fr}\\[6pt]

\normalsize

\end{center}

% Here is the abstract:

\noindent
Many discrete optimization problems amount to selecting a feasible \blue{set of edges} of least weight. We consider in this paper the context of spatial graphs where the positions of the vertices are uncertain and belong to known uncertainty sets. The objective is to minimize the sum of the distances of the chosen \blue{set of edges} for the worst positions of the vertices in their uncertainty sets. We first prove that these problems are $\cal NP$-hard even when the feasible sets consist either of all spanning trees or of all $s-t$ paths. Given this hardness, we propose an exact solution algorithm combining integer programming formulations with a cutting plane algorithm, identifying the cases where the separation problem can be solved efficiently. We also propose a conservative approximation and show its equivalence to the affine decision rule approximation in the context of Euclidean distances. We compare our algorithms to three deterministic reformulations on instances inspired by the scientific literature for the Steiner tree problem and a facility location problem.

\bigskip


\noindent\hrulefill

% begin ``double spacing" the text:
\baselineskip 20pt plus .3pt minus .1pt

%%%%%%%%%%%%%%%%%%%%%%%%%%%%%%%%%%%%%%%%%%%%%%%%%%%%%%%%%%%%%%%%%%%%%%

% Samples of sectioning (and labeling) in OPRE
% NOTE: (1) \section and \subsection do NOT end with a period
%       (2) \subsubsection and lower need end punctuation
%       (3) capitalization is as shown (title style).
%
%\section{Introduction.}\label{intro} %%1.
%\subsection{Duality and the Classical EOQ Problem.}\label{class-EOQ} %% 1.1.
%\subsection{Outline.}\label{outline1} %% 1.2.
%\subsubsection{Cyclic Schedules for the General Deterministic SMDP.}
%  \label{cyclic-schedules} %% 1.2.1
%\section{Problem Description.}\label{problemdescription} %% 2.

% Text of your paper here

\section{Introduction}
Research in combinatorial optimization has provided efficient algorithms to solve many complex discrete decision problems, providing exact or near-optimal solutions in reasonable amounts of time. The applications are countless, ranging from logistics (network design, facility location, $\ldots$) to scheduling. \blue{In this paper, we are interested in the class $\PBFAMILY$ of deterministic combinatorial optimization problems that amount to selecting a feasible set of edges in a given graph $\GG=(\VV,\EE)$ and that minimizes the sum of edge-weights. Any $\Pi\in \PBFAMILY$ represents a specific problem, such as the shortest path or the minimum spanning tree problem. We consider further that $\GG$ is a spatial graph embedded into a given metric space $(\Metric, d)$. Each vertex~$i$ is assigned a position~$u_i\in\Metric$ so the weight of each edge $\{i,j\}$ is given by its distance $d(u_i,u_j)$. Denoting by $\X\subseteq\{0,1\}^{|\EE|}$ the set of feasible vectors for a given instance, any $\Pi\in\PBFAMILY$ corresponds to a combinatorial optimization problem of the form
\begin{equation}
\tag{$\Pi$}
\label{eq:CO}
 \min_{x\in \X}\sum_{\{i,j\}\in \EE} x_{ij }d(u_i,u_j).
\end{equation}
}
Problem~\ref{eq:CO} encompasses many applications, such as network design and facility location. %, and clustering. 
These are typically subject to data uncertainty, be it because of the duration of the decision process, measurement errors, or simply lack of information. 
One successful framework that has emerged to address uncertainty is robust optimization~\citep{ben1998robust}, where the uncertain parameters are modeled with convex sets such as polytopes, or with finite sets of points.
\blue{Many authors have focused more particularly on robust discrete optimization problems see~\cite{BertsimasS03,BuchheimK18,kasperski2016robust,kouvelis2013robust} and the references therein.} 
We enter this framework by considering the model where the positions of the vertices are subject to uncertainty, therefore impacting the distances among the vertices. The resulting problem thus seeks to find \blue{a} feasible set of edges that minimizes its worst-case sum of distances. Formally, we introduce for each vertex $i\in \VV$ the set of possible locations as the uncertainty set $\U_i%=\{u_i^1,\ldots,u_i^{\nU{i}}\}
\subseteq \Metric$ of cardinality $\nU{i}=|\U_i|$. \blue{We consider that there is no correlation between the positions of the different vertices, so a scenario is given by the tuple $\u=(u_1,\ldots,u_{|\VV|})$ which belongs to the set $\UU=\times_{i\in \VV}\U_i$. Then, given $\Pi\in\PBFAMILY$, we study in this paper the \emph{locational robust counterpart} of problem $\Pi$, formally defined as
\begin{equation}
\tag{\ROBUST{$\Pi$}}
\label{eq:minmaxCO}
\min_{x\in\X} \max_{\u\in\UU}\sum_{\{i,j\}\in \EE}x_{ij} d(u_i,u_j).
\end{equation}
}
We also devote a particular attention to evaluating the objective function of~\ref{eq:minmaxCO}, \blue{often called the adversarial problem}
\blue{
\begin{equation}
\tag{\EVALC}
\label{eq:evalc}
\max_{\u\in\UU}\sum_{\{i,j\}\in \EE}x_{ij} d(u_i,u_j).
\end{equation}
We underline that we focus throughout on finite uncertainty sets. However, our setting encompasses polyhedral uncertainty sets whenever the distance function is convex.
\begin{remark}
\label{rem:convex}
  Suppose that $\Metric\subseteq \R^p$ for some $p>0$ and that $d$ is a convex function. Then, maximizing over $\UU$ is equivalent to maximizing over the polytope $\conv(\UU)=\times_{i\in \VV}\conv(\U_i)$. Hence, in that case our setting covers polyhedral uncertainty sets.
\end{remark}
}
As an illustration, the following two applications fall into the context of problem~\ref{eq:minmaxCO}.
\begin{example}[Subway network design] Designing and expanding a subway network forms an important optimization problem faced by large cities. The new lines should efficiently cover dense city areas while interacting well with the existing transportation lines. A key aspect of this problem amounts to locating the new subway stations. In addition to the technical considerations inherent to any construction, these also involve political considerations as local officials are never happy to let their citizens face the inconvenience of heavy civil engineering. This political lever is particularly complex in cities like Brussels having multiple local governments that must all agree before the stations can finally be constructed. As the overall process takes years, facing local government changes, the exact locations of the metro stations typically evolve between the first draft and the final implementation. Now, the exact locations of the stations impact the lengths of the resulting lines, the construction cost of which is typically proportional to their lengths~\citep{GUTIERREZJARPA20133000}. The cost of digging the new lines can therefore be modeled as a network design problem with locational uncertainty on the position of the vertices, usually including additional technical and environmental constraints. %When a single line is of interest, $\G$ will consist of \emph{paths}, \emph{cycles}, or \emph{trees}.
\end{example}
\begin{example}[Strategic facility location]
\label{ex:UFL}
A production company wishes to expand its activities in a new region, locating additional facilities. We consider the strategic level where the company may only choose  approximate locations, as the exact locations will be known later, after all technical and legal considerations have been studied. \blue{We consider a one-stage location problem where the selection of facilities and assignment of customers are decided at the same time. This is relevant, for instance, when clients may need different types of products that require different installations at the facilities.} 
As always in such facility location problems, the distances between the future clients and facilities lead to significant transportation costs that need to be kept as low as possible. In this particular case, these distances depend on locations that are uncertain at the time planning decisions are made. %This leads to the problem of covering the vertices representing the clients with a certain number of \emph{disjoint stars}, the root of each star representing a facility. 
Importantly, the distances are provided by the underlying road network~\citep{MELKOTE2001515}, which yields a graph-induced metric $(\Metric,d)$. \blue{More generally, this framework is also relevant for any application of the $p$-median problem~\citep{Marin2019} with locational uncertainty.}% This includes data clustering problems~\citep{kohn2010p}.
\end{example}
%\begin{example}[Data clustering]
%\label{ex:clustering}
%Given a set of data points, a fundamental problem in classification seeks to partition the points into a given number of subsets so as to minimize the dissimilarities among the points grouped into the same subset. The data points may consist of real values, for instance, when they model physical measurements, or ordinal values, representing the results of medical or psychological surveys. In both cases, uncertainty on the values is common, be it because of a measurement error, or because of a lack of information -- in which case the corresponding coordinate is replaced by the full interval~\citep{masson2020cautious}. One must thus group the data in a way that is robust against these uncertainties. Various objective functions and metric spaces may be considered. With real values, it is common to minimize the sum of squared Euclidean distances. With ordinal values instead, a linear function may be more appropriate, coupled with ad hoc graph-induced distances.
%\end{example}

Traditionally, robust optimization problems with an objective function that is concave in the uncertain parameters are reformulated as compact models using conic duality~\citep{ben1998robust}. These techniques do not readily extend to function $d(u_i,u_j)$ as the latter is non-concave in general. Actually, for Euclidean metric spaces based on the vector space $\R^\ell, \ell\in\mathbb{Z}^+$, $d(u_i,u_j)=\|u_i-u_j\|_2$ is convex in $u_i$ and $u_j$. Function $\|u_i-u_j\|_2$ is closely related to the second-order cone (SOC) constraints considered by~\cite{zhen2021robust} for robust problems with polyhedral uncertainty sets. \cite{zhen2021robust} linearize such robust SOC constraints by introducing adjustable variables, turning the problem into an adjustable robust optimization problem that can be tackled exactly~\citep{AyoubP16,ZhenHS18,zeng2013solving} or approximately using affine decision rules~\citep{Ben-TalGGN04} or finite adaptability approaches~\citep{BertsimasD16,hanasusanto2015k,postek2016multistage,subramanyam2019k}, among others. 

Interestingly, the approach of~\cite{zhen2021robust}, extended in~\cite{roos2018approximation} to more general convex functions, makes no particular assumption on the feasibility set of the decision variables $\X$. A second work closely related to~\ref{eq:minmaxCO} is that of~\cite{citovsky2017tsp}, who rely on computational geometry techniques to provide constant-factor approximation algorithms in the special case where $\X$ contains all Hamiltonian cycles of $\GG$. They propose in particular to solve a deterministic counterpart of~\ref{eq:minmaxCO} where the uncertain distances are replaced by the maximum pairwise distances $\dmax_{ij}=\max_{u_i\in \U_i, u_j\in \U_j} \dist{u_i}{u_j}$, for each $(i,j)\in \VV^2, i\neq j$.

%Given the limitations of compact reformulations, another line of algorithms has emerged for difficult robust optimization problems with convex uncertainty.  These algorithms replace the convex uncertainty set by a finite approximation, leading to a relaxation of the original problem. Then, they iterate between solving integer programming formulations for the robust problem with finite uncertainty, and checking the solution's optimality to the relaxation by solving an adversarial separation problem. These approaches have been used not only for simple versions of static robust optimization~\citep{BertsimasDL16,FischettiM12}, but also in more complex problems having multiple decision stages. Initiated by~\cite{BienstockO08} and~\cite{zeng2013solving}, these algorithms have made possible to numerically solve to near-optimality a wide range of robust discrete optimization problems, such as robust inventory routing problems~\citep{AgraSNP16} or node selection problems~\citep{Naoum-SawayaB16}. They have also been adapted by~\cite{subramanyam2019k} to address adjustable robust optimization involving discrete recourse decisions.

To summarize, we see that while~\cite{zhen2021robust} provide valuable tools for addressing problems defined in Euclidean metric spaces considering uncertainty polytopes, their approaches cannot be used for graph-induced metric spaces, such as those mentioned in Example~\ref{ex:UFL}. % and~\ref{ex:clustering}\todo{Correct this depending on the choice.}. 
On the other hand,~\cite{citovsky2017tsp} focused on the case where $\X$ contains all Hamiltonian cycles of $\GG$. 
The main purpose of the present paper is thus to provide a more general solution algorithm that is valid for any set $\X$ and metric space $(\Metric,d)$. 
We only assume that $\UU$ is finite, encompassing the two cases mentioned above. \blue{Specifically, polyhedral uncertainty in Euclidean metric spaces is already discussed in Remark~\ref{rem:convex}.} Then, in the case of graph-induced metrics, the set $\Metric$ is the set of nodes of a finite graph, meaning that each $\U_i\subseteq \Metric$ must be finite as well.

\blue{Let us denote by $G(x)=(V(x),E(x))$ the subgraph induced by $x$, where $E(x)=\set{\{i,j\}\in \EE}{x_{ij}=1}$ and $V(x)=\set{i\in \VV}{\exists e \in E(x):i\in e}$}.
In this context, we can summarize our contributions as follows:
\begin{itemize}
\item We prove that ~\ref{eq:minmaxCO} is \NPH even when \blue{$\X$} consists of all $s-t$ paths and $(\Metric,d)$ is the one-dimensional Euclidean metric space or when \blue{$\X$} consists of all spanning trees of $\GG$. These results illustrate how the nature of~\ref{eq:minmaxCO} fundamentally differs from the classical min-max robust problem with cost uncertainty, which is known to be polynomially solvable whenever the costs lie in independent uncertainty sets~\citep{AissiBV09}.
 
\item We provide a general cutting-plane algorithm for~\ref{eq:minmaxCO}. We further show that problem \ref{eq:evalc} is \NPH and provide two algorithms for computing \ref{eq:evalc}. One is based on integer programming formulations while the other one relies on a dynamic programming algorithm that involves the threewidth of \blue{$G(x)$}.

\item We leverage the above dynamic programming to provide a compact formulation for the problem when any \blue{$G(x)$} contains only stars (or unions of stars). We can, in theory, extend that idea to trees, albeit presenting poor numerical performance.

\item We propose a conservative approximation of the problem that uncouples $\UU$ into its projections $\U_i$, $i\in \VV$. In the case of Euclidean metric spaces, this approximation leads to mixed-integer second-order conic reformulations, and turns out to be equivalent to the affine decision rule reformulation proposed by~\cite{zhen2021robust}.
 
\item We compare the exact cutting plane algorithm numerically with the above conservative approximation and simple deterministic reformulations. The benchmark is composed of two families of instances. 
%\todo{Add clustering?}
The first family includes Steiner tree instances that illustrate subway network design. 
The second one is composed of strategic facility location instances.
%The former application relies on two-dimensional Euclidean metric spaces so we can further include the affine decision rule reformulation from~\cite{zhen2021robust} to the comparison.
\end{itemize}

The rest of the paper is structured as follows. Section~\ref{sec:NPhard} studies the hardness of~\ref{eq:evalc} and~\ref{eq:minmaxCO}. In Section~\ref{sec:exact}, we develop the exact solution algorithm. %, including the detail of the computations of $\cost(G)$. 
The latter involves a dynamic programming algorithm for trees, generalized to graphs of bounded threewidth in the appendix. Section~\ref{sec:adr} details the conservative reformulation.
In Section~\ref{sec:num}, we present our numerical experiments. The appendix details the extension of the dynamic programming to graphs with bounded threewidth (Appendices~\ref{app:fptandtw} and~\ref{app:tw2}), the details of the compact formulations (Appendix~\ref{app:compactSTP}) for trees, and the equivalence between our conservative reformulation and that of~\cite{zhen2021robust} (Appendix~\ref{app:adr}).

%%%%%%%%%%%%%%%%%%%%%%%%%%%%%%%%%%%%%%%%%%%%%%%%%%%%%%%%%%%%%%%%%%%%%%%%%%%%%%%%%%%%%%%%%%%%%%%%%%%%%%%%%%%
\section{Hardness results}
\label{sec:NPhard}
%%%%%%%%%%%%%%%%%%%%%%%%%%%%%%%%%%%%%%%%%%%%%%%%%%%%%%%%%%%%%%%%%%%%%%%%%%%%%%%%%%%%%%%%%%%%%%%%%%%%%%%%%%%

We study in this section the complexity of the optimization problems \ref{eq:minmaxCO} and \ref{eq:evalc}. 

\subsection{Problem \textsc{robust}-$\Pi$}

\blue{Let us start by observing that \ref{eq:minmaxCO} is not harder than its nominal counterpart \ref{eq:CO} whenever the edges indexed by each $x\in\X$ are disjoint (do not have common endpoints). In that case, we have
$$
\max_{\u\in\UU}\sum_{\{i,j\}\in \EE}x_{ij} d(u_i,u_j) = \sum_{\{i,j\}\in \EE} \max_{u_i\in \U_i,u_j\in \U_j} x_{ij}\dist{u_i}{u_j}=\sum_{\{i,j\}\in \EE}x_{ij}\dmax_{ij}.
$$
Hence, solving \ref{eq:minmaxCO} in the above setting amounts to solve
$$
 \min_{x\in \X}\sum_{\{i,j\}\in \EE} x_{ij} \dmax_{ij}.
$$
As an example, consider that \ref{eq:CO} is the assignment problem so $\GG$ is a bi-partite graph based on the partition of $\VV$ into the two sets $\VV^1$ and $\VV^2$ of equal size, and any $x\in\X$ selects $|\VV^1|$ edges that cover all nodes. Problem \ref{eq:CO} being solvable in polynomial time, so is \ref{eq:minmaxCO}.}

In spite of this easy example, we show in this section that \ref{eq:minmaxCO} can in general not be reduced to \ref{eq:CO} as \ref{eq:minmaxCO} is typically harder than \ref{eq:CO}. We illustrate the hardness of \ref{eq:minmaxCO} by focusing on the two well-known polynomially solvable problems \ref{eq:CO}, namely the shortest path problem \SP and the minimum spanning tree problem \MST. These problems have been largely studied in the robust combinatorial optimization literature under cost uncertainty (e.g.~\cite{kasperski2009approximability,YamanKP01}), \blue{which addresses problems of the type
\begin{equation}
\label{eq:minmax}
\min_{x\in \X} \max_{c\in\C}\sum_{\{i,j\}\in \EE} x_{ij}c_{ij},
\end{equation}
where $\C$} is a given uncertainty set included in the positive orthant. \blue{Comparing~\ref{eq:minmax} with~\ref{eq:minmaxCO} underlines that the difficulty of~\ref{eq:minmaxCO} lies in the non-linearity of the distance function $d$.}
 It is folklore (e.g., \cite{AissiBV09}) that when $\C$ is the Cartesian product of intervals, $\C=\times_{e \in \EE}[\underline{c}_e,\overline{c}_e]$ for $\underline{c}_e \leq \overline{c}_e$, problem~\ref{eq:minmax} can be reformulated as $\min_{x\in \X} \sum_{\{i,j\}\in \EE} x_{ij}\overline{c}_{ij}$ making the robust problems as easy as their nominal counterparts. Our first result below shows that such is not the case for \ref{eq:minmaxCO}, as \SP turns \NPH even in the simple case where each $\U_i$ is a subset of $\R$. Notice that in the 1-dimensional Euclidean space, the convexity of $d(u_i,u_j)=| u_i - u_j|$ implies that $\U_i$ is equivalent to the set $\{\underline{u}_i,\overline{u}_i\}$, for some $\underline{u}_i\leq \overline{u}_i$, \blue{in line with Remark~\ref{rem:convex}.}% so Proposition~\ref{prop:SPNPhard} illustrates how \ROBUST{SP}  is hard even in that simple case.

\begin{proposition}
  \label{prop:SPNPhard}
  \ROBUST{SP} is \NPH even when $(\Metric,d)$ is the 1-dimensional Euclidean space.
\end{proposition}
\begin{proof}
 Given a set of integers $\{a_1,\ldots,a_n\}$, with $A=\sum_{i=1}^n a_i$, the $\cal{NP}$-complete decision problem partition asks for a subset $S\subset\{1,\ldots,n\}$ such that $\sum_{i\in S} a_i = A/2$. Let $K>0$ be a large enough integer. The reduction considers the graph $\GG$ with $2n+2$ vertices and $4n$ edges as illustrated Figure~\ref{fig:SPNPhardSeg}; the regions $\U_i$ are translated away from vertex $o$ for visibility. Specifically, our reduction locates vertices $s$ and $t$ at $0$ while \blue{$\U_i=\{-\un_i,\up_i\}$} for each vertex $i$ different from $s$ and $t$. The definition of $\up$ and $\un$ alternates along the vertices $v_i,v_{i+1},v_{i+2},\ldots$ and similarly for vertices $w_i$: for each $i=2k+1$, we define $\up_{v_i}=K+a_i$, $\un_{v_i}=K+\frac{A}{n}-a_i$, $\up_{w_i}=K$ and $\un_{w_i}=K+\frac{A}{n}$, while for each $i=2k$, we define $\up_{v_i}=K+\frac{A}{n}-a_i$, $\un_{v_i}=K+a_i$, $\up_{w_i}=K+\frac{A}{n}$ and $\un_{w_i}=K$.

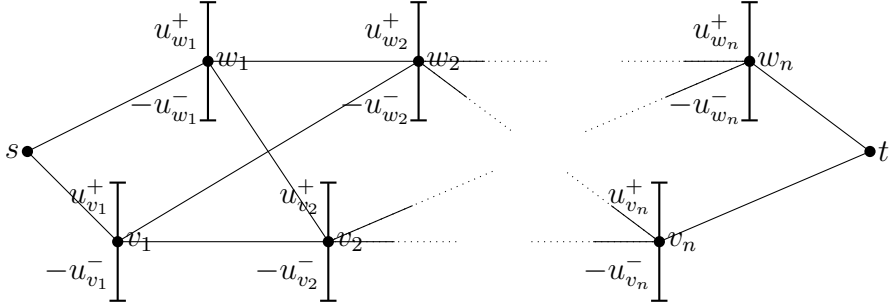
\begin{figure}[!h]
\centering
\scriptsize
\begin{tikzpicture}[scale=0.8]
\tkzDefPoint(0,0){o}
\tkzLabelPoint[left](o){$s$}
\tkzDefPoint(14,0){d}
\tkzLabelPoint[right](d){$t$}

\tkzDefPoint(1.5,-1.5){v1}
\tkzLabelPoint[right](v1){$v_1$}
\tkzDefPoint(1.5,-2){v1n}
\tkzLabelPoint[left](v1n){$-\un_{v_1}$}
\tkzDefPoint(1.5,-0.75){v1u}
\tkzLabelPoint[left](v1u){$\up_{v_1}$}
\draw[thick,{|-|}] (1.5,-2.5) -- (1.5,-0.5);

\tkzDefPoint(3,1.5){w1}
\tkzLabelPoint[right](w1){$w_1$}
\tkzDefPoint(3,0.75){w1n}
\tkzLabelPoint[left](w1n){$-\un_{w_1}$}
\tkzDefPoint(3,2){w1u}
\tkzLabelPoint[left](w1u){$\up_{w_1}$}
\draw[thick,{|-|}] (3,0.5) -- (3,2.5);

\tkzDefPoint(5,-1.5){v2}
\tkzLabelPoint[right](v2){$v_2$}
\tkzDefPoint(5,-2){v2n}
\tkzLabelPoint[left](v2n){$-\un_{v_2}$}
\tkzDefPoint(5,-0.75){v2u}
\tkzLabelPoint[left](v2u){$\up_{v_2}$}
\draw[thick,{|-|}] (5,-2.5) -- (5,-0.5);

\tkzDefPoint(6.5,1.5){w2}
\tkzLabelPoint[right](w2){$w_2$}
\tkzDefPoint(6.5,0.75){w2n}
\tkzLabelPoint[left](w2n){$-\un_{w_2}$}
\tkzDefPoint(6.5,2){w2u}
\tkzLabelPoint[left](w2u){$\up_{w_2}$}
\draw[thick,{|-|}] (6.5,0.5) -- (6.5,2.5);

\tkzDefPoint(10.5,-1.5){vn}
\tkzLabelPoint[right](vn){$v_n$}
\tkzDefPoint(10.5,-2){vnn}
\tkzLabelPoint[left](vnn){$-\un_{v_n}$}
\tkzDefPoint(10.5,-0.75){vnu}
\tkzLabelPoint[left](vnu){$\up_{v_n}$}
\draw[thick,{|-|}] (10.5,-2.5) -- (10.5,-0.5);

\tkzDefPoint(12,1.5){wn}
\tkzLabelPoint[right](wn){$w_n$}
\tkzDefPoint(12,0.75){wnn}
\tkzLabelPoint[left](wnn){$-\un_{w_n}$}
\tkzDefPoint(12,2){wnu}
\tkzLabelPoint[left](wnu){$\up_{w_n}$}
\draw[thick,{|-|}] (12,0.5) -- (12,2.5);

\foreach \n in {o,d,v1,w1,v2,w2,wn,vn}
  \node at (\n)[circle,fill,inner sep=1.5pt]{};

\foreach \v/\w in {o/w1, o/v1, w1/w2, w1/v2, v1/w2, v1/v2, wn/d, vn/d}
    \draw (\v) to (\w);

\foreach \v/\w in {v2/vn, v2/wn, w2/vn, w2/wn, vn/v2, vn/w2, wn/v2, wn/w2}
    \draw (\v) edge ($(\v)!0.2!(\w)$) edge [dotted] ($(\v)!0.4!(\w)$);

\end{tikzpicture}
\caption{\label{fig:SPNPhardSeg} Reduction from partition when $\UU$ is a Cartesian product of segments.}
\end{figure}

We first show that for $K$ large enough, the worst-case $\u\in\UU$ for any path $P$ from $s$ to $t$ alternates from the top of an interval to the bottom of the subsequent interval along the path. To prove this, notice that for any vertex $v\in \VV \setminus\{s,t\}$, $\un_v\in[K-A,K+A]$ and $\up_v\in[K-A,K+A]$ and the same holds for any vertex $w$. Hence, if $\u$ alternates for the entire path, the resulting cost is not smaller than $c=2n(K-A)$. On the contrary, if $\u$ misses one alternation, its cost cannot be greater than $c'=2(n-1)(K+A)+2A$. Hence, taking $K > 2nA$ ensures $c > c'$. 

The reduction works as follow. Let $S\subseteq\{1,\ldots,n\}$ be a subset of integers and $\overline{S}$ its complement. We associate to $S$ the path $P_S$ from $s$ to $t$ that contains $v_i$ for each $i\in S$ and $w_i$ for each $i\in\overline{S}$. From the above, only two scenarios in $\UU$ must be considered in the worst-case and each vertex $i\in \{w_1,v_1,\ldots,w_n,v_n\}$ contributes to the total length with either $2\up_i$ or $2\un_i$, depending on the scenario considered. We have
\begin{align*}
c(P_S) &= \max_{\u\in\UU}\sum_{\{i,j\}\in P_S} \|u_i-u_j\|_2 \\
&=2\max\left(nK+\sum_{i\in S}a_i,n(K+\frac{A}{n})-\sum_{i\in S}a_i\right) \\
&=2\max\left(nK+\sum_{i\in S}a_i,A + nK-\sum_{i\in S}a_i\right) \\
&=2nK + 2\max\left(\sum_{i\in S}a_i,\sum_{i\in \overline{S}}a_i\right).
\end{align*}
Hence, there exists a path $P_S$ in $\X$ with minimum cost of $2nK+A$ if and only if there exists a set $S$ such that $\sum_{i\in S} a_i = \sum_{i\in \overline{S}} a_i=A/2$.
\end{proof}

For \MST we can prove the hardness of the problem only for a more general metric space.
\begin{proposition}
\label{prop:MSThard}
\ROBUST{MST}  is \NPH.
\end{proposition}
\begin{proof}
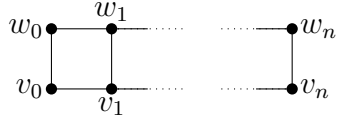
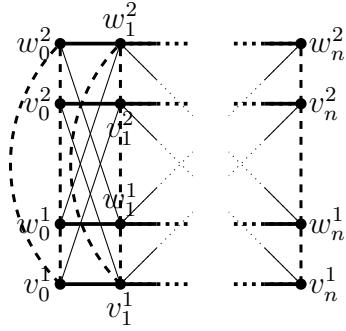
\begin{figure}
\centering
\scriptsize
\subfloat[$\GG$]{
\begin{tikzpicture}[scale=0.8]
\tkzDefPoint(0,0){v0}
\tkzLabelPoint[left](v0){$v_{0}$}
\tkzDefPoint(0,1){w0}
\tkzLabelPoint[left](w0){$w_{0}$}

\tkzDefPoint(1,0){v1}
\tkzLabelPoint[below](v1){$v_{1}$}
\tkzDefPoint(1,1){w1}
\tkzLabelPoint[above](w1){$w_{1}$}

\tkzDefPoint(4,0){vn}
\tkzLabelPoint[right](vn){$v_n$}
\tkzDefPoint(4,1){wn}
\tkzLabelPoint[right](wn){$w_n$}

\foreach \n in {v0,v1,vn,w0,w1,wn}
  \node at (\n)[circle,fill,inner sep=1.5pt]{};

\foreach \v/\w in {v0/w0, v1/w1, vn/wn, v0/v1, w0/w1}
    \draw[-] (\v) to (\w);
    
\foreach \v/\w in {v1/vn, vn/v1, w1/wn, wn/w1}
    \draw (\v) edge ($(\v)!0.2!(\w)$) edge [dotted] ($(\v)!0.4!(\w)$);
\end{tikzpicture}
\label{fig:MSTNPhard:G}
}

\subfloat[$G_{\Metric}$.]{
\begin{tikzpicture}[scale=0.8]
\tkzDefPoint(0,0){v01}
\tkzLabelPoint[left](v01){$v_{0}^1$}
\tkzDefPoint(0,1){w01}
\tkzLabelPoint[left](w01){$w_{0}^1$}

\tkzDefPoint(1,0){v11}
\tkzLabelPoint[below](v11){$v_{1}^1$}
\tkzDefPoint(1,1){w11}
\tkzLabelPoint[above](w11){$w_{1}^1$}

\tkzDefPoint(4,0){vn1}
\tkzLabelPoint[right](vn1){$v_n^1$}
\tkzDefPoint(4,1){wn1}
\tkzLabelPoint[right](wn1){$w_n^1$}

\tkzDefPoint(0,3){v02}
\tkzLabelPoint[left](v02){$v_{0}^2$}
\tkzDefPoint(0,4){w02}
\tkzLabelPoint[left](w02){$w_{0}^2$}

\tkzDefPoint(1,3){v12}
\tkzLabelPoint[below](v12){$v_{1}^2$}
\tkzDefPoint(1,4){w12}
\tkzLabelPoint[above](w12){$w_{1}^2$}

\tkzDefPoint(4,3){vn2}
\tkzLabelPoint[right](vn2){$v_n^2$}
\tkzDefPoint(4,4){wn2}
\tkzLabelPoint[right](wn2){$w_n^2$}

\foreach \n in {v01,v11,vn1,w01,w11,wn1,v02,v12,vn2,w02,w12,wn2}
  \node at (\n)[circle,fill,inner sep=1.5pt]{};

\foreach \v/\w in {v01/v11, v02/v12, w01/w11, w02/w12}
    \draw[line width=0.5mm] (\v) to (\w);
    
\foreach \v/\w in {v01/w01, v11/w11, vn1/wn1, v02/w02, v12/w12, vn2/wn2, w01/v02, w11/v12, wn1/vn2}
    \draw[dashed, line width=0.4mm] (\v) to (\w);
    
\foreach \v/\w in {v01/v12, v02/v11, w01/w12, w02/w11}
    \draw[-] (\v) to (\w);

\foreach \v/\w in {v01/w02, v11/w12, vn1, wn2}
    \draw[dashed, line width=0.4mm] (\v) to[bend left=45] (\w);
    
\foreach \v/\w in {v11/vn1, vn1/v11, w11/wn1, wn1/w11, v12/vn2, vn2/v12, w12/wn2, wn2/w12}
    \draw[line width=0.5mm] (\v) edge ($(\v)!0.2!(\w)$) edge [dotted] ($(\v)!0.4!(\w)$);
    
\foreach \v/\w in {v11/vn2, vn2/v11, w11/wn2, wn2/w11, w12/wn1, wn1/w12, v12/vn1, vn1/v12}
    \draw (\v) edge ($(\v)!0.2!(\w)$) edge [dotted] ($(\v)!0.4!(\w)$);
\end{tikzpicture}
\label{fig:MSTNPhard:GM}
}
\caption{\label{fig:MSTNPhard} Graphs used in the reduction for the minimum spanning tree problem.}
\end{figure}

 We consider the same partition problem as in the proof of Proposition~\ref{prop:SPNPhard}. Now $\GG$ contains the $2n+2$ vertices and $3n+1$ edges depicted on Figure~\ref{fig:MSTNPhard:G}. We consider the metric space $(\Metric,d)$ induced by the weighted graph $G_{\Metric}=(V_{\Metric},E_{\Metric},\omega)$ depicted on Figure~\ref{fig:MSTNPhard:GM}. Let $K>0$ be a number large enough. The dashed edges and thin edges have their weights equal to $K$ and $2K$, respectively, while $\omega_{v_{i-1}^1v_{i}^1}=3K+a_i$, $\omega_{v_{i-1}^2v_{i}^2}=3K+\frac{A}{n}-a_i$, $\omega_{w_{i-1}^1w_{i}^1}=3K$, and $\omega_{w_{i-1}^2w_{i}^2}=3K+\frac{A}{n}$ for each $i=1,\ldots,n$. The weight vector $\omega$ satisfies the triangle inequalities, so the metric $d$ induced on $V_{\Metric}$ by the shortest paths in $G_{\Metric}$ satisfies $d_{ij}=\omega_{ij}$ for each $\{i,j\}\in E_{\Metric}$.
 Finally, we define $\U_{v_i}=\{v_i^1,v_i^2\}$ and $\U_{w_i}=\{w_i^1,w_i^2\}$ for $i=0,\ldots,n$.
 
We first observe that the cost of a vertical edge $\{v_i,w_i\}$ is equal to $K$ for all positions of $(v_i,w_i)\in\U_{v_i}\times \U_{w_i}$. Let us consider any tree $T$ in $\GG$ that contains $n^v$ vertical edges and $n^h$ horizontal edges, where $n \leq n^h\leq 2n$. For $K$ large enough, we claim that the worst-case $\u\in\U$ locates all vertices either in the bottom layer of $G_{\Metric}$ that consists of vertices $v^1_i$ and $w^1_i$ for  $i=0,\ldots,n$, or in the top layer that consists of the remaining vertices. To prove the claim, notice that the weight of any horizontal edge in $G_{\Metric}$ is comprised between $3K-A$ and $3K+A$, while the weight of any diagonal edge is $2K$. Hence, if $\u$ locates all its vertices either in the bottom or in the top layer, the resulting cost is not smaller than $c=n^h(3K-A)+n^vK$. On the contrary, if $\u$ alternates at least once between the layers, its cost cannot be greater than $c'=(n^h-1)(3K+A)+2K+n^vK$. Hence, taking $K>(4n-1)A\geq(2n^h-1)A$ ensures $c>c'$, proving the claim.

We prove next that for $K$ large enough, any optimal tree $T$ in $\GG$ must contain $n + 1$ vertical edges and $n$ horizontal ones. Following the above claim, the cost of a horizontal edge $\{v_i,v_{i+1}\}$ or $\{w_i,w_{i+1}\}$ for a worst-case $\u\in\U$ is comprised between $3K-A$ and $3K+A$. Hence, any tree $T$ with $n^h\in\{n+1,\ldots,2n\}$ horizontal edges costs at least $c=n^h(3K-A)+n^vK$ while any tree having $n^h-1$ horizontal edges costs at most $c=(n^h-1)(3K+A)+(n^v+1)K$. Hence, taking $K>2nA\geq n^hA$ ensures $c>c'$, proving $n^h=n$ in any optimal solution. 

As in the proof of Proposition~\ref{prop:SPNPhard}, we let $S\subseteq\{1,\ldots,n\}$ be a subset of integers and $\overline{S}$ its complement. We associate to $S$ the tree $T_S$ that contains $\{v_{i-1},v_i\}$ for each $i\in S$ and $\{w_{i-1},w_i\}$ for each $i\in\overline{S}$. Following the claim above, only two scenarios in $\UU$ must be considered, and following again the reasoning used in the proof of Proposition~\ref{prop:SPNPhard}, we have
$$
c(T_S) = \max\left(3nK+\sum_{i\in S}a_i,3nK+A-\sum_{i\in S}a_i\right) 
=3nK + \max\left(\sum_{i\in S}a_i,\sum_{i\in \overline{S}}a_i\right).
$$
Hence, there exists a spanning tree $T_S$ in $\X$ with minimum cost of $3nK+A/2$ if and only if there exists a set $S$ such that $\sum_{i\in S} a_i = \sum_{i\in \overline{S}} a_i=A/2$. 
\end{proof}

We detail in the remark below how, for any positive integer $\ell$, the metric space $(\Metric,d)$ used in the proof of Proposition~\ref{prop:MSThard} cannot be embedded isometrically into $\R^\ell$. 
As a consequence, the hardness of \ROBUST{MST}  in Euclidean spaces remains an open problem.

\begin{remark}
 The graph $G_{\Metric}$ described in the above proof cannot be embedded isometrically into an Euclidean space, as can be seen by considering the triangle $w_0^1\, w_{1}^1\, w_0^2$ and the fourth point $v_0^2$. The sides of the triangle have length $d(w_0^1,w_{1}^1)=3K$, $d(w_{1}^1,w_0^2)=2K$, and $d(w_0^2,w_0^1)=2K$. Hence, since $d(v_0^2,w_0^1)=d(v_0^2,w_0^2)=K$, any isometric embedding maps $v_0^2$ to the midpoint of segment $w_0^1\,w_0^2$, so its Euclidean distance to $w_{1}^1$ must be $\sqrt{\frac{11}{2}}K$. This is in contradiction with $d(v_0^2,w_{1}^1)=\min(\omega_{v_0^2w_0^1}+\omega_{w_0^1w_{1}^1},\omega_{v_0^2v_1^2}+\omega_{v_1^2w_{1}^1})=\min(4K,4K+\frac{A}{n}-a_1)$. The above illustrates that when $\X$ contains all spanning trees of $\GG$, the complexity of~\ref{eq:minmaxCO}  is still open when one considers only Euclidean metric spaces.
\end{remark}

%Having proved that \ref{eq:minmaxCO} is \NPH, we will develop two types of solution algorithms for the problem. In the next section, we provide a cutting plane algorithm based on integer programming formulations for $\F$, which has an exponential running-time in general. 

\subsection{Problem $\EVALC$}
\label{sec:EVALChard}

We now turn to the difficulty of computing the objective function \ref{eq:evalc}. \blue{Given that $x$ is fixed throughout, we denote $G(x)$ more shortly as $G$. Furthermore, we denote the sets of vertices and edges of $G$ as $V[G]$ and $E[G]$, respectively.} Our first result (Proposition~\ref{prop:maxcut} below) is that \ref{eq:evalc} is hard, even when the metric space is reduced to two points, or the input graph is a clique. For this, we consider particularly simple metric spaces, and rely on a reduction from problem \MAXCUT. We recall that \MAXCUT is a famous problem in combinatorial optimization that, given any input graph $G$, seeks a partition $\{V_1,V_2\}$ of $V[G]$ such that $|\{e \in E[G] : \card{e \cap V_1}=1\}|$ is maximized.

\begin{proposition}\label{prop:maxcut}
Even when $|\Metric|=2$, there is no \PTAS for \ref{eq:evalc} unless \PEQNP.
\end{proposition}
\begin{proof}
Let us denote the objective function of \MAXCUT as $f^{\MAXCUT}(V_1,V_2)=|\{e \in E[G]\mid |e \cap V_1|=1\}|$. Further, we denote by $\OPT^{\MAXCUT}(G)$ the value of an optimal solution for graph $G$. Given an input graph $G$ of \MAXCUT, we define $\Metric=\{0,1\}$ and $I$ as the graph $G$ itself, $\U_i = \Metric$ for any $i \in V[G]$, and the distance $d$ by $d(x,y)=|x-y|$.

Given a solution $\{V_1,V_2\}$ (which is a partition) of \MAXCUT, we define $u_i = 0$ if $i \in V_1$, and $1$ if $i \in V_2$. This implies $c(\u,I) = f^{\MAXCUT}(V_1,V_2)$. For the reverse direction, given a solution $u$ of \ref{eq:evalc}, we define $V_1 = \{i \mid u_i=0\}$ and $V_2 = V[G] \setminus V_1$, and we also have $c(\u,I) = f^{\MAXCUT}(V_1,V_2)$.

The above immediatly implies that there is an $S$-reduction (see for instane~\cite{612321}) from \MAXCUT to \ref{eq:evalc}, proving the result.
\end{proof}

% Our second result assumes instead that $G$ is a clique, and relies on a reduction from problem \MAXDIVERSITY. Given a ground set $X$, a natural number $k$ and a diversity function $div:2^X\to\mathbb{R}$, \MAXDIVERSITY seeks to output a subset $S \subseteq X$ of size $k$ that maximizes $div(S)$. The special case of interest here, proved \NPH in~\cite{DBLP:journals/corr/abs-1809-09521}, considers that the ground set $X$ is the three-dimensional Euclidean space, and that $div(S)=\sum_{i,j \in X}\|i-j\|_2$. We obtain immediately the following result.
%  
% \begin{proposition}\label{prop:divsersity}
%    Even when restricted to input graphs that are cliques, to metric space defined by squared Euclidean distance in the three-dimensional Euclidean space,   and to instances where $\U_i = \Metric$ for all $i$, \ref{eq:evalc} is \NPH.
%  \end{proposition}

%\JO{Marin: For what follows, we need the definition of treewidth and of \EVALC/$tw$, maybe even explaining the purpose of studying this problem. We will also need some reminder of what W[1]-hard is. And about the notations, I think that the that the calligraphic fonts is devoted to sets and complexity classes. So FPT algorithm should come with normal font, don't you think?}
%\mp{I partly answered this, by providing definitions. But the other questions remain, in particular, motivating with we would like to focus on \EVALC/$tw$.}
Let us now turn to parameterized complexity,  and let $tw$ be the treewidth of $G$, see Appendix~\ref{app:tw} for the formal definition of treewidth. Informally, $tw$
measures the thickness of a tree structure defining $G$.
In particular, $tw(G)=1$ for any tree $G$.
As we show in the next section that computing \ref{eq:evalc} is polynomial on trees, a natural
question is to determine if we can extend this result by proving that \EVALC/$tw$ admits an \FPT algorithm (where \EVALC/$tw$ denotes problem \ref{eq:evalc} parameterized by $tw$, as defined
in Appendix~\ref{app:fpt}), meaning an algorithm running in $f(tw)\cdot |I|^c$ for some computable function $f$ and constant $c$.
The following proposition implies that it is very unlikely, and thus places \ref{eq:evalc} with the few problems that are not \FPT by treewidth. 
\begin{proposition}\label{prop:w1tw}
\EVALC/tw is \WONEH.
\end{proposition}
\begin{proof}
  Given a graph $G$, and a set of integers (called colors) $L(i)$ for any $i \in V[G]$, problem \LISTCOL aims at deciding whether we can find a color $f(i) \in L(i)$ for any $i \in V[G]$ such that for any edge $\{i,j\} \in E[G]$, $f(i) \neq f(j)$.
  It is known~\cite{fellows2011complexity} that \LISTCOL/$tw$ is \WONEH. Let us now prove that there is a parameterized reduction from \LISTCOL/$tw$ to \EVALC/$tw$,
  which implies (see Appendix~\ref{app:fpt}) that \EVALC/$tw$ is \WONEH.

Given a graph $G$ a list of colors $L(i)$ for any $i \in V[G]$, we define $\Metric=\bigcup_{i \in V[G]}L(i)$, and $d(c_1,c_2)=0$ is $c_1=c_2$, and $1$ otherwise.
We define the uncertainty set of $G$ as follows: for any $i$, we let $\U_i = L(i)$.
It is now straightforward to verify that we have a YES-instance of \LISTCOL if and only if $\cost(G)=\blue{|V[E]|}$. As the reduction can be computed in polynomial time, and the graph (and thus its treewidth) is unchanged, this provides a parameterized reduction, and we get the desired result.
\end{proof}

%%%%%%%%%%%%%%%%%%%%%%%%%%%%%%%%%%%%%%%%%%%%%%%%%%%%%%%%%%%%%%%%%%%%%%%%%%%%%%%%%%%%%%%%%%%%%%%%%%%%%%%%%%%
\section{Exact solution of \textsc{robust}-$\Pi$}
\label{sec:exact}
%%%%%%%%%%%%%%%%%%%%%%%%%%%%%%%%%%%%%%%%%%%%%%%%%%%%%%%%%%%%%%%%%%%%%%%%%%%%%%%%%%%%%%%%%%%%%%%%%%%%%%%%%%%

A popular type of algorithms solving exactly difficult robust optimization problems replaces the large uncertainty set by an approximation of small cardinality, leading to a relaxation of the original problem. Then, these algorithms iterate between solving integer programming formulations for the robust problem with small uncertainty set, and checking the optimality of the solution for the relaxation by solving an adversarial separation problem. This process leads to cutting plane algorithms (e.g.,~\cite{BertsimasDL16,FischettiM12,Naoum-SawayaB16}). Such algorithms involve frequent calls to computing the objective function \ref{eq:evalc}, so we start this section by studying how to solve this problem. Then, we detail in Section~\ref{sec:CProbustPi} the overall cutting plane algorithm for \ref{eq:minmaxCO}. %These approaches have been used not only for simple versions of static robust optimization~\citep{BertsimasDL16,FischettiM12}, but also in more complex problems having multiple decision stages. Initiated by~\cite{BienstockO08} and~\cite{zeng2013solving}, these algorithms have made possible to numerically solve to near-optimality a wide range of robust discrete optimization problems, such as robust inventory routing problems~\citep{AgraSNP16} or node selection problems~\citep{Naoum-SawayaB16}. They have also been adapted by~\cite{subramanyam2019k} to address adjustable robust optimization involving discrete recourse decisions.

\subsection{\blue{Problem \EVALC}}
\label{subsec:tw}

%Let us consider throughout the graph induced by set of edges $G\in\G$, namely $G(F)=(V[G],E[G])$. 
\blue{As in Section~\ref{sec:EVALChard}, we denote $G(x)$ more shortly as $G$ in what follows.}
Given the hardness results from the previous section, we propose two approaches to computing \ref{eq:evalc} that have non-polynomial running times in general. The first approach relies on an integer programming formulation. For each $i\in \VV$ and $k\in\{1,\ldots,\nU{i}\}$, binary variable $y_i^k$ takes value 1 if and only if vertex $i$ is located at position $u_i^k$. Therefore, \ref{eq:evalc} is equal to
\begin{align*}
 \max & \sum_{\{i,j\}\in E[G]} \sum_{k=1}^{\nU{i}}\sum_{\ell=1}^{\nU{j}} d(u_i^k,u_j^\ell)y_i^k y_j^\ell\\
 \mbox{s.t.} & \sum_{k=1}^{\nU{i}} y_i^k =1, \quad \forall i\in V[G] \\
 & y_i\in\{0,1\}^{\nU{i}},\quad\forall i\in V[G]
\end{align*}
which can be linearized using classical techniques.

It is also possible to compute \ref{eq:evalc} efficiently whenever $G$ has small treewidth $tw(G)$ using a dynamic programming algorithm. Let us detail the algorithm whenever $G$ is a tree rooted at vertex $r$, which we assume oriented from $r$ to its leaves $L$. We denote by $D(i)$ the set that contains the direct descendants of $i$, which is empty if $i$ is a leaf. Let $\OPT(i,u_i)$ be the maximum value obtained for the subtree starting at $i$ given that node $i$ is located at $u_i$. We obtain the following recursion:
\begin{equation}
 \label{eq:basicDP}
\OPT(i,u_i) = \left\{
\begin{array}{ll}
\sum\limits_{j \in D(i)} \max\limits_{u_j \in \U_j} d(u_i,u_j) + \OPT(j,u_j), & i\in V[G]\setminus L \\
0, & i \in L
\end{array}
\right.
\end{equation}
and the optimal solution cost is given by $\max_{u_r\in\U_r}\OPT(r,u_r)$. Dynamic programming recursion~\eqref{eq:basicDP} will be used in our numerical experiments, which involve trees and stars.

%\mp{Marin: Jeremy wanted more details on what "classical techniques" means. Here comes my suggestion.}
Recall that $tw = tw(G)$ and let us further denote $\NU=\max_{i \in V[G]}\nU{i}$. Using dynamic programming on a well-chosen tree decomposition of $G$ (see Appendix~\ref{app:tw} for the definition), one can readily extend the above idea to any graph of bounded treewidth, leading to Theorem~\ref{thm:tw}, whose proof is deferred to Appendix~\ref{app:tw2}. We point out that according to Proposition~\ref{prop:w1tw} we cannot (unless $\mathcal{W}[1]=$\FPT) remove the dependency in $\NU$ to get for example a $\grandO(poly(n) \times f(tw))$, and this holds for any computable function~$f$.

\begin{theorem}\label{thm:tw}
  \EVALC/$tw+\NU$ is \FPT. More precisely, we can compute an optimal solution of \ref{eq:evalc} in time $\grandO(n \times tw \times \NU^{\grandO(tw)})$.
\end{theorem}

\subsection{Cutting plane algorithm for the robust problem}
\label{sec:CProbustPi}

Now that we have depicted numerical methods for computing \ref{eq:evalc}, we wish to make the extra step towards the exact solution of the complete problem, \ref{eq:minmaxCO}.
For this, we design an exact solution algorithm that generates scenarios of $\UU$ on the fly in the course of a branch-and-cut algorithm. 

Let $\tUU$ be a finite subset of $\UU$. An exact algorithm for~\ref{eq:minmaxCO}, described in Algorithm~\ref{algo:cp}, relies on the following relaxed formulation
\begin{equation}
\label{eq:ILPrelaxed}
\min\set{\omega}{\omega \geq \sum_{\{i,j\}\in \EE} x_{ij} d(u_i,u_j), \; \forall \u\in\tUU, x \in \X}.
\end{equation}

Algorithm~\ref{algo:cp} describes an iterative cutting-plane implementation, alternating between the solution of the relaxed master problem~\eqref{eq:ILPrelaxed} and the adversarial separation problem \ref{eq:evalc}. Practical implementation of these algorithms typically rely instead on branch-and-cut algorithms, where the adversarial separation problem is solved at each integer node of the branch-and-bound-tree. 

\begin{algorithm}
\DontPrintSemicolon
\SetKwInOut{Return}{return}
\SetKwInOut{Initialization}{initialization}
\Repeat{$\cost(G) \leq \tomega$}
{
Let $(\tomega,\tx)$ be an optimal solution of~\eqref{eq:ILPrelaxed}\;
Let $G$ be the graph induced by $\tx$\;
Compute $\cost(G)=\max_{\u\in\UU}\sum_{\{i,j\}\in E[G]} d(\tu_i,\tu_j)$ and let $\tu$ be a maximizer\;
\lIf{$\cost(G) > \tomega$}{$\tUU \leftarrow \tUU \cup \{\tu\}$}
}
\textbf{return} $G$
\caption{Cutting-plane algorithm for~\ref{eq:minmaxCO}}
\label{algo:cp}
\end{algorithm}

%For large and complex problems, one can hardly expect Algorithm~\ref{algo:cp} to run quickly, so it may be wiser to first try a quick approximation algorithm. We provide in the next section such an algorithm and study in depth the worst-case bound between the solution cost returned by the algorithm and the optimal solution cost.

\subsection{Compact formulation for stars}
\label{sec:compact}

Depending on the structure of the elements of $\X$, the dynamic programming recursion~\eqref{eq:basicDP} naturally leads to a compact formulation for the problem. We detail next this idea for the case where any $x\in\X$ describes a union of disjoint stars rooted at the elements of a known set $\roots\subseteq \VV$. 
%$$\F=\set{F=\bigcup_{\root\in\roots}F_\root}{\mbox{$F_\root$ is a star rooted at $\root$}}$$
For each $\root\in\roots$ and $u_\root^k\in\U_\root$, let us introduce the optimization variable $z_{\root}^k$ to model $\OPT(\root,u_\root^k)$, the cost of the star rooted at $\root$ given that $u_\root=u_\root^k$.
Let $N(i)$ be the set of neighbours of any node $i\in\VV$. Plugging variables $z$ and $x$ into~\eqref{eq:basicDP} and noticing that any node connected to $\root$ must be a leaf, we immediately obtain
\begin{equation}
\label{eq:zrootmax}
 z_{\root}^k = \sum\limits_{j \in N(\root)} x_{\root j} \max\limits_{u_j \in \U_j} d(u_\root^k,u_j).
\end{equation}
Notice that the maximization appearing in the right-hand-side of~\eqref{eq:zrootmax} does not involve optimization variables, se we can define the constant $\dmax_{\root k j}=\max_{u_j \in \U_j} d(u_\root^k,u_j)$. Then, introducing $z_\root$ as the worst-case cost of the star rooted at $\root$, we have
\begin{equation}
\label{eq:zrootdef}
z_{\root} = \max\limits_{k\in [\nU{\root}]}z_{\root}^k = \max\limits_{k\in [\nU{\root}]}\sum\limits_{j \in N(\root)} x_{\root j} \dmax_{\root k j}. 
\end{equation}
Overall, we wish to minimize the sum of $z_{\root}$ over all $\root\in\roots$.  Reformulating~\eqref{eq:zrootdef} through an epigraphic reformulation, we obtain
\begin{align}
\min\quad & \omega \nonumber\\
\mbox{s.t.}\quad & \omega \geq \sum_{\root \in \roots} z_\root \\
& z_\root \geq \sum\limits_{j \in N(\root)} x_{\root j} \dmax_{\root k j}, \quad \forall \root\in\roots, k\in [\nU{\root}]\\
& x \in \X, z\geq 0\nonumber
\end{align}
The above construction can, in theory, be extended to trees. However, in that case the recurrence relations lead to products between variables, which turns out to be inefficient numerically. See Appendix~\ref{app:compactSTP} for details. 

%%%%%%%%%%%%%%%%%%%%%%%%%%%%%%%%%%%%%%%%%%%%%%%%%%%%%%%%%%%%%%%%%%%%%%%%%%%%%%%%%%%%%%%%%%%%%%%%%%%%%%%%%%%
\section{Conservative approximation}
\label{sec:adr}

We introduce next a (conservative) approximation of~\ref{eq:minmaxCO} that leads to compact formulations. Let us introduce an additional optimization variable $\mu_e\in\Metric$ for each $e\in \EE$, and consider the following optimization problem 
\begin{equation}
\tag{\CONSERVATIVE{$\Pi$}}
\label{eq:minmaxCOmu}
\min_{x\in \X \atop \mu\in\Metric^{|\EE|}} \max_{\u\in\UU}\sum_{\{i,j\}\in \EE} x_{ij}(\dist{u_i}{\mu_{ij}}+\dist{\mu_{ij}}{u_j}).
\end{equation}
\blue{One might interpret the additional variable $\mu_{ij}$ as a compulsory crossing point from vertex $i$ to vertex $j$, regardless of the position of these vertices. Using these crossing points, each distance function only involves a single node, leading to simpler reformulations as we show below.}
\begin{remark}
 Due to the triangle inequalities, the optimal solution cost of~\ref{eq:minmaxCOmu} is not smaller than the optimal solution cost of~\ref{eq:minmaxCO}, so~\ref{eq:minmaxCOmu} is a conservative approximation of~\ref{eq:minmaxCO}.
\end{remark}

We show next how to reformulate~\ref{eq:minmaxCOmu} as a discrete optimization problem featuring a polynomial number of variables. Noticing that 
$$
\sum_{\{i,j\}\in \EE} x_{ij}(\dist{u_i}{\mu_{ij}}+\dist{\mu_{ij}}{u_j})
=
\sum_{i\in \VV} \sum_{\{i,j\}\in \EE} x_{ij}\dist{u_i}{\mu_{ij}},
$$
we obtain
$$
\max_{\u\in\UU}\sum_{\{i,j\}\in \EE} x_{ij}(\dist{u_i}{\mu_{ij}}+\dist{\mu_{ij}}{u_j})
=
\max_{\u\in\UU}\sum_{i\in \VV} \sum_{\{i,j\}\in \EE} x_{ij}\dist{u_i}{\mu_{ij}}
=
\sum_{i\in \VV} \max_{u_i\in\U_i}\sum_{\{i,j\}\in \EE} x_{ij}\dist{u_i}{\mu_{ij}}.
$$
Thus, we can introduce an additional variable $d_i$ for each node $i\in \VV$, so~\ref{eq:minmaxCOmu} can be reformulated as
 \begin{align}
\min\quad & \sum_{i\in \VV}d_i \label{eq:adr_general_first}\\
\mbox{s.t.}\quad & d_i \geq \sum_{\{i,j\}\in \EE} x_{ij}\dist{u_i}{\mu_{ij}}, \quad\forall i\in V,u_i\in \U_i\label{eq:adr_general}\\
& x \in \X, \;\mu\in\Metric^{|\EE|}. \label{eq:adr_general_last}
\end{align}
The interest of the above reformulation is that the uncertainty sets $\U_i$, $i\in\VV$, appear in distinct constraints, so~\eqref{eq:adr_general} contains $\sum_{i\in \VV}\nU{i}$ constraints, which is significantly smaller than the $\prod_{i\in \VV}\nU{i}$ elements in the global uncertainty set $\UU$. 
In practice, the numerical difficulty of problem~\eqref{eq:adr_general_first}--\eqref{eq:adr_general_last} typically depends on the considered metric space~$(\Metric,d)$ and feasibility set~$\X$. For instance, using ad-hoc pre-processing rules, we may be able to reduce the domain of each variable $\mu_{ij}$ to a small subset of $\Metric_{ij}\subset\Metric$. These rules may not even need to be exact as problem~\eqref{eq:adr_general_first}--\eqref{eq:adr_general_last} is only a conservative approximation of the original problem~\ref{eq:minmaxCO}.
%Yet, these constraints involve the additional discrete variables $\mu_e$, $e\in E$, each of them taking any value in the set $\Metric$, leading to optimization problems that may in general difficult to solve from the numerical viewpoint.

In what follows, we further develop the case where $(\Metric,d)$ is the $p$-dimensional Euclidean space so the distance $\dist{u_i}{u_j}=\|u_i-u_j\|_2$ is now well-defined for any $u_i,u_j\in \R^p$. We can leverage this to relax the discrete restriction $\mu\in\Metric^{|\EE|}$ to $\mu\in\R^{p\times|\EE|}$, obtaining
\begin{align}
\min\quad & \sum_{i\in \VV}d_i \\
\mbox{s.t.}\quad & d_i \geq \sum_{\{i,j\}\in \EE} x_{ij}\|u_i-\mu_{ij}\|_2, \quad\forall i\in \VV,u_i\in \U_i\label{eq:adr_euclidean}\\
& x \in \X, \; \mu\in\Metric^{|\EE|}.
\end{align}
The non-linearities in constraints~\eqref{eq:adr_euclidean} can be avoided by replacing $x_{ij}\|u_i-\mu_{ij}\|_2$ with $\|x_{ij} u_i-\mu_{ij}\|_2$: if $x_{ij}=1$, both expressions coincide; otherwise, $x_{ij}=0$, and setting $\mu_{ij}=0$ also yields equality. Introducing additional variables to isolate each norm into a unique constraint, we finally obtain the following mixed-integer second-order cone programming formulation
\begin{align}
\min\quad &  \sum_{i\in \VV}d_i \label{adr:first}\\
\mbox{s.t.}\quad &
d_i \geq \sum_{e\in \EE}\sum_{i\in e} \nu_{i,e}^k\\
& \nu_{i,e}^k \geq \|x_{e} u_i^k-\mu_{e}\|_2, \quad \forall e\in \EE, i\in e, k\in[\nU{i}]\\
& x \in \X,\;  \mu\in\Metric^{|\EE|}.\label{adr:last}
\end{align}

We conclude this section by mentioning that~\eqref{adr:first}--\eqref{adr:last} can be alternatively obtained by following the approach proposed in~\citet{zhen2021robust}. Specifically, let us recall the epigraphic reformulation of~\ref{eq:minmaxCO}
\begin{align}
\min\quad & \omega \label{eq:eculideanfirst}\\
\mbox{s.t.}\quad & \omega \geq \sum_{\{i,j\}\in \EE} x_{ij} \|u_i-u_j\|_2, \quad \forall \u\in\UU\label{eq:epigraphicL2}\\
& x \in \X.\label{eq:eculideanlast}
\end{align}
We detail in Appendix~\ref{app:adr} how each constraint~\eqref{eq:epigraphicL2} can be reformulated by introducing recourse variables which, approximated through affine decision rules, leads exactly to \eqref{adr:first}--\eqref{adr:last}. This connection underlines that the difference between the optimal solution costs of~\ref{eq:minmaxCO} and~\ref{eq:minmaxCOmu} can be interpreted as the suboptimality of affine decision rules for approximating two-stage robust optimization. It also suggests that stronger conservative approximations could be obtained by using more expressive decision rules, such as the lifted affine decision rules proposed by~\cite{de2017tractable}.

%%%%%%%%%%%%%%%%%%%%%%%%%%%%%%%%%%%%%%%%%%%%%%%%%%%%%%%%%%%%%%%%%%%%%%%%%%%%%%%%%%%%%%%%%%%%%%%%%%%%%%%%%%%
\section{Computational experiments}
\label{sec:num}
%%%%%%%%%%%%%%%%%%%%%%%%%%%%%%%%%%%%%%%%%%%%%%%%%%%%%%%%%%%%%%%%%%%%%%%%%%%%%%%%%%%%%%%%%%%%%%%%%%%%%%%%%%%

In this section, we compare numerically the exact algorithm from Section~\ref{sec:exact}, denoted \exact{} hereafter, with three heuristic algorithms that solve deterministic counterparts of~\ref{eq:minmaxCO}. Namely, each of these heuristics considers a symmetric function $\dd:V\times V\rightarrow\R_+$ and returns the optimal solution of $\min_{x\in \X}\sum_{\{i,j\}\in \EE} x_{ij}\dd_{ij}$. Three such functions $\dd$ are considered:
\begin{description}
 \item[\worst{}:] $\dd_{ij}=\max\limits_{u_i\in \U_i, u_j\in \U_j} \dist{u_i}{u_j}$, as suggested by~\cite{citovsky2017tsp}.
 \item[\hcenter{}:] $\dd_{ij}=\dist{\bc_i}{\bc_j}$, where $\bc_i$ is any \emph{geometric median} of $\U_i$, defined as $$\bc_i\in\argmin\limits_{u\in\Metric}\sum\limits_{u'\in \U_i}\dist{u}{u'}.$$
%  Notice that in the case of Euclidean distance, the above medoid is unique and coincides with the barycenter of the set.}
 \item[\avg{}:] $\dd_{ij}=\frac{1}{\nU{i}\nU{j}}\sum\limits_{u_i\in\U_i,u_j\in\U_j}\dist{u_i}{u_j}$.
\end{description}
We also include in our numerical assessment the conservative approximation depicted in Section~\ref{sec:adr}, and denoted \cons{}. We compare these algorithms on the applications mentioned in the introduction: a subway network design problem, modeled as a Steiner tree problem (STP), and a simple plant location problem (SPL). 
Since the applications involve stars and trees, the separation problems of \exact{} can be solved using the dynamic programming recurrence presented in~\eqref{eq:basicDP}.

The purpose of our experiments is two-fold. First and foremost, we wish to assess the numerical efficiency of the exact solution algorithm in terms of solution times. Second, we measure the approximation ratios obtained by the heuristic algorithms, by comparing the cost of their solutions to the optimal solution costs.
%\todo{Shouldn't we make a reference to a companion article that includes theoretical approximation ratios for \worst{} and \avg{}? Otherwise, it might seem odd to talk only about practical approximation ratios. Or maybe do we wait for the reviewers comments so that the companion article is actually finalized and can be referenced?}

The algorithms have been coded in Julia~\citep{bezanson2012julia}, using JuMP~\citep{DunningHuchetteLubin2017} to interface the mixed integer linear programming (MILP) solver CPLEX. They have been carried out on a processor Intel(R) Core(TM) i7-10510U CPU\@1.80GHz using up to 4 threads in parallel and with a total running time limit of 2 hours. The source code of every algorithm is publicly available at \url{https://github.com/mjposs/locational\_uncertainty}.

\subsection{Steiner tree problem}

We consider the problem of expanding the subway network of a city, modeled as a Euclidean Steiner tree problem. The compulsory points model the future stops of the subway, while the other points model the possible knickpoints of the lines. 
%Designing such an expansion plan is a complex urban planning problem that involve multiple political and economical layers. In particular, even when the topological design of the network has been decided, the exact physical location of the stations and knickpoints may still change, be it because of political constraints (inability to buy the rights of a given location) or physical ones (impossibility to dig as planned). The lengths of the lines are directly impacted by these uncertainties, and so is their cost, which can reasonably be assumed proportional to their Euclidean lengths.
We thus consider an undirected graph $\GG=(\VV,\EE)$ where $T\subseteq \VV$ denotes the set of compulsory vertices; we consider an arbitrary root $r\in T$ and set $T_0=T\setminus\{r\}$. 
Set $\X$ thus contains all trees of $\GG$ that cover the vertices of $T$.
Sets $\U_i\subseteq\R^2$ model the possible locations for the vertices, which we assume to be polyhedral sets, and we assume that the distance $d(u_i,u_j)=\|\u_i-\u_j\|_2$ is the Euclidean distance. 

We consider the classical disaggregated MILP formulation for the problem involving two sets of variables~\citep{magnanti1984network}. 
For each undirected edge $\{i,j\}\in \EE$, binary variable $x_{ij}$ takes value 1 if and only if the edge is used. Then, for each $t\in T_0$ and $e=\{i,j\}\in \EE$, the fractional variable $f^t_{ij}$ decides how much flow related to $t$ is sent on the directed arc $(i,j)$. 
Let $\EE^{bidir}$ be the set of directed edges obtained from $\EE$ by including the two opposite edges $(i,j)$ and $(j,i)$ for each undirected edge $\{i,j\}\in \EE$. Defining the incoming and outgoing stars at node $i$ as $\delta^-(i)=\set{j}{(j,i)\in \EE^{bidir}}$ and $\delta^+(i)=\set{j}{(i,j)\in \EE^{bidir}}$, respectively, and the balance of vertex $i$ as $b^t_i=0$ for $i\in T_0\setminus\{t\}$, $b^t_{r}=-1$ and $b^t_t=1$, we obtain
\begin{align*}
\min & \left(\max_{\u\in \UU} \sum_{\{i,j\}\in \EE}x_{ij}\|u_i-u_j\|_2\right) \\
\mbox{s.t.} &\sum_{(j,i)\in \delta^-(i)} f^t_{ji} - \sum_{(i,j)\in \delta^+(i)} f^t_{ij} = b^t_i,\quad\forall i\in \VV, t\in T_0 \\
& f_{ij}^t+f_{ji}^t \leq x_{ij},\quad \forall \{i,j\}\in \EE, t\in T_0 \\
& f\geq0,\; x\text{ binary} 
\end{align*}
\subsubsection{Instances}

We assess the different solution algorithms on the instances P6E with 100 vertices and 5 terminals (p619, p620, and 621) that are publicly available at \url{http://steinlib.zib.de/testset.php}. Each of these instances has 180 edges. The position of the vertices, denoted $\bar{u}_i$ hereafter, are not available in the data files P6E, so we estimate them using a variant of the MDS-MAP algorithm from~\cite{shang2003localization}. 
Specifically, we apply classical multidimensional scaling (MDS) from the Julia package MultivariateStats (see \url{https://github.com/JuliaStats/MultivariateStats.jl}) to compute the positions $\bar{u}$ from the distances, completing the distance matrix with the shortest path values.
The uncertainty sets $\U_i, i\in \VV$ are then computed randomly based on two parameters: $\diaminst$ that scales the diameter of each set $\U_i$, and $\NU$ the common number of elements of all $\U_i, i\in \VV$. 
To be more precise, we first compute the average distance among pairs of points in $\VV$, $\bar{d}=\sum_{i< j}\frac{\|\bar{u}_i-\bar{u}_j\|}{n(n-1)/2}$. 
For each $i\in\VV$, we then uniformly draw one random value in $\rho_i\in [0,\diaminst \cdot\bar{d}]$ and define the circle $C_i$ of center $\bar{u}_i$ and radius $\rho_i$. Then, we take $\NU$ equidistant points on $C_i$, yielding
$$
\U_i=\left\{\left(\bar{u}_{i1}+\rho_i\cos\left(\frac{2 k\pi}{\NU}\right),\bar{u}_{i2}+\rho_i\sin\left(\frac{2 k\pi}{\NU}\right)\right), k=1,\ldots,\NU\right\}.
$$
\blue{We consider each $\Delta\in\{0.2,0.4,0.6\}$ and $\sigma\in\{4,8,12\}$.}
 Following the above procedure, we create 5 random instances for each P6E instance and choice of parameters, yielding 135 instances in total. 
% As the Chebyshev center XXXXXX that for these experiments we use $\bc_i$ as the center of circle $C_i$, $\bar{u}_i$.

\subsubsection{Results}

Figure~\ref{fig:STP:times:P6E} reports the average solution times, illustrating the impact of the dimension of the diameters of the uncertainty sets, represented by $\Delta$, and the number of elements in each set, given by $\NU$. The figure illustrates that, unsurprisingly, the three deterministic counterparts are solved much faster than \cons{} and \exact{}. More interesting is the fact that \exact{} is faster than the heuristic algorithm \cons{} when $\Delta$ is small. However, the difficulty of solving \exact{} grows rapidly with the value of $\Delta$. Notice also that three instances based on p621, corresponding to $\Delta=0.6$, could not be solved to exact optimality within two hours, ending with optimality gaps of 1\%, 5\%, and 7\%, respectively. Hence, the value reported on Figure~\ref{fig:STP:times:P6E:Delta} for $\Delta=0.6$ is actually a lower bound for the true average value.

\begin{figure}[ht]
\centering
\subfloat[Varying $\Delta$\label{fig:STP:times:P6E:Delta}]{
\begin{tikzpicture}[scale=0.75]
\begin{axis}[height=5cm,width=6cm,axis x line=bottom, axis y line = left,ymin = 0.0001,
	ylabel={\blue{solution times (seonds)}},
	title={},
	%	ymode = log,
    log ticks with fixed point,
 	every axis plot/.append style={ultra thick}]
	
	\addplot [solid, blue, mark=triangle] table[
	x=x, y=center
	]{Steiner/P6E_times_Delta.txt};
	
	\addplot [dotted, red, mark=o, mark options=solid] table[
	x=x, y=worst
	]{Steiner/P6E_times_Delta.txt};
	
	\addplot [dashed, teal, mark=square, mark options=solid] table[
	x=x, y=avg
	]{Steiner/P6E_times_Delta.txt};
	
	\addplot [solid, black, mark=triangle] table[
	x=x, y=cons
	]{Steiner/P6E_times_Delta.txt};
	
	\addplot [dotted, brown, mark=o, mark options=solid] table[
	x=x, y=exact
	]{Steiner/P6E_times_Delta.txt};
\end{axis}
\end{tikzpicture}
}
\subfloat[Varying $\sigma$]{
\begin{tikzpicture}[scale=0.75]
\begin{axis}[height=5cm,width=6cm,axis x line=bottom, axis y line = left,ymin = 0.0001,
	ylabel={\blue{solution times (seonds)}},
	title={},
	%	ymode = log,
    log ticks with fixed point,
 	every axis plot/.append style={ultra thick},
 	legend entries={\hcenter, \worst, \avg, \cons, \exact},
	legend style={at={(axis cs:13,50)}, anchor=south west}]
	
	\addplot [solid, blue, mark=triangle] table[
	x=x, y=center
	]{Steiner/P6E_times_nU.txt};
	
	\addplot [dotted, red, mark=o, mark options=solid] table[
	x=x, y=worst
	]{Steiner/P6E_times_nU.txt};
	
	\addplot [dashed, teal, mark=square, mark options=solid] table[
	x=x, y=avg
	]{Steiner/P6E_times_nU.txt};
	
	\addplot [solid, black, mark=triangle] table[
	x=x, y=cons
	]{Steiner/P6E_times_nU.txt};
	
	\addplot [dotted, brown, mark=o, mark options=solid] table[
	x=x, y=exact
	]{Steiner/P6E_times_nU.txt};
\end{axis}
\end{tikzpicture}
}
\caption{STP: Average solution times in seconds on instances P6E for each algorithm when varying one of the parameters.}
\label{fig:STP:times:P6E}
\end{figure}
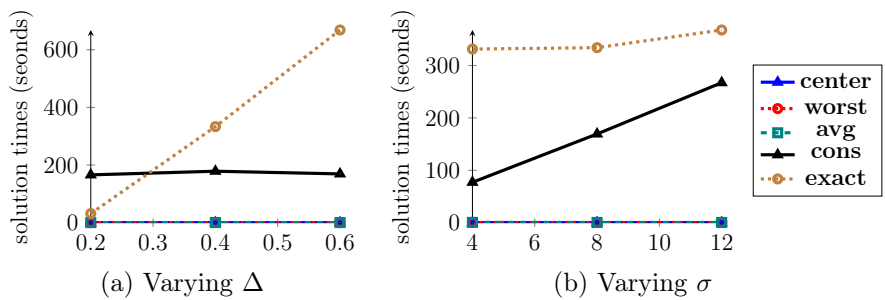

\begin{figure}[ht]
\centering
\subfloat[\scriptsize P6E, $\diaminst=0.2$]{
\begin{tikzpicture}[scale=0.7]
\begin{axis}[height=5cm,width=6cm,axis x line=bottom, axis y line = left,xmin = 0,ymin = 0.001, ymax = 100,
xmax = 50,
	ylabel={\blue{\% of instances}},
	title={},
    log ticks with fixed point,
 	every axis plot/.append style={ultra thick}]
	
	\addplot [solid, blue] table[
	x=x, y=solved, mark=none
	]{Steiner/P6E_center_0.2.txt};
	
	\addplot [dotted, red] table[
	x=x, y=solved, mark=none
	]{Steiner/P6E_worst_0.2.txt};
	
	\addplot [dashed, teal] table[
	x=x, y=solved, mark=none
	]{Steiner/P6E_avg_0.2.txt};
	
	\addplot [solid, black] table[
	x=x, y=solved, mark=none
	]{Steiner/P6E_cons_0.2.txt};
\end{axis}
\end{tikzpicture}
}
\subfloat[\scriptsize P6E, $\diaminst=0.4$]{
\begin{tikzpicture}[scale=0.7]
\begin{axis}[height=5cm,width=6cm,axis x line=bottom, axis y line = left,xmin = 0,ymin = 0.001, ymax = 100,
xmax = 50,
	ylabel={\blue{\% of instances}},
	title={},
    log ticks with fixed point,
every axis plot/.append style={ultra thick}]
	
	\addplot [solid, blue] table[
	x=x, y=solved, mark=none
	]{Steiner/P6E_center_0.4.txt};
	
	\addplot [dotted, red] table[
	x=x, y=solved, mark=none
	]{Steiner/P6E_worst_0.4.txt};
	
	\addplot [dashed, teal] table[
	x=x, y=solved, mark=none
	]{Steiner/P6E_avg_0.4.txt};
	
	\addplot [solid, black] table[
	x=x, y=solved, mark=none
	]{Steiner/P6E_cons_0.4.txt};
\end{axis}
\end{tikzpicture}
}
\subfloat[\scriptsize P6E, $\diaminst=0.6$]{
\begin{tikzpicture}[scale=0.7]
\begin{axis}[height=5cm,width=6cm,axis x line=bottom, axis y line = left,xmin = 0,ymin = 0.001, ymax = 100,
xmax = 50,
	ylabel={\blue{\% of instances}},
	title={},
    log ticks with fixed point,
    legend entries={\hcenter,\worst,\avg,\cons},
 	legend style={at={(axis cs:40,25)}, anchor=south west},
	every axis plot/.append style={ultra thick}]
	
	\addplot [solid, blue] table[
	x=x, y=solved, mark=none
	]{Steiner/P6E_center_0.6.txt};
	
	\addplot [dotted, red] table[
	x=x, y=solved, mark=none
	]{Steiner/P6E_worst_0.6.txt};
	
	\addplot [dashed, teal] table[
	x=x, y=solved, mark=none
	]{Steiner/P6E_avg_0.6.txt};
	
	\addplot [solid, black] table[
	x=x, y=solved, mark=none
	]{Steiner/P6E_cons_0.6.txt};
\end{axis}
\end{tikzpicture}
}
\caption{STP: \% of instances for which the additional relative cost is less than $x$.}
\label{fig:STP:costs}
\end{figure}
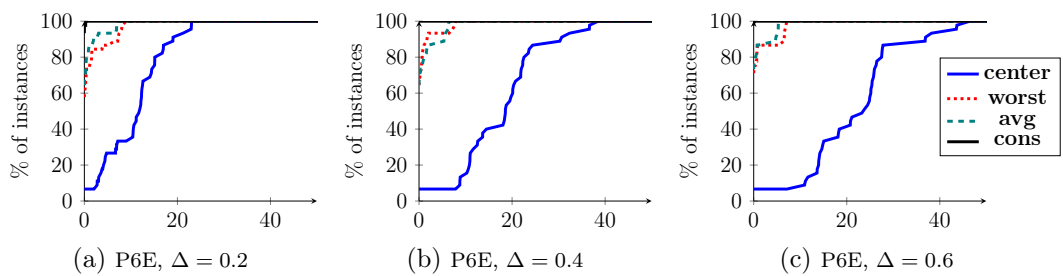

%More specifically, Figure~\ref{fig:STP:times:small} highlights the rapidly increasing solution times of \cons{} along with the values of $\NU$ and $\kappa$. The three other approaches, including \exact{}, can solve all these instances within a fraction of second. Figure~\ref{fig:STP:times:large} further depicts how \exact{} is sensitive to $\diaminst$ for the larger instances, while \hcenter{} and \worst{} seem less affected by the value of that parameter.

Figure~\ref{fig:STP:costs} reports the cumulative distributions of the cost increase for each of the four heuristic algorithms, relatively to the cost of the exact solution. Formally, let $z(H)$ denote the cost of the solution returned by $H\in\{\hcenter{},\worst{},\avg{},\cons{}\}$ and $z^*$ denote the optimal solution cost. 
For each $H\in\{\hcenter{},\worst{},\avg{},\cons{}\}$, the corresponding curve reports
\begin{equation}
 \label{eq:costcurves}
g(x) = 100\frac{\#\{\mbox{instances for which $z(H)\leq (1+x)\cdot z^*$}\}}{\#\{\mbox{all instances}\}}.
\end{equation}
%Notice \cons{} is only assessed on the small instances as that approach is unable to cope with the larger ones in reasonable solution times. 
These results show that \worst{}, \avg{} and \cons{} provide solutions with values very close to the optimal one, with \cons{} being the best of the three, always leading to the optimal solution. In contrast, the quality of \hcenter{} becomes rather poor as $\diaminst$ increases, ranging up to an extra cost 50\% for some of these instances, and with nearly half of the instances with $\diaminst=0.6$ having an extra cost of at least 20\%. 

\begin{figure}
 \centering
\subfloat[\scriptsize \exact{}]{
\includegraphics[width=8cm]{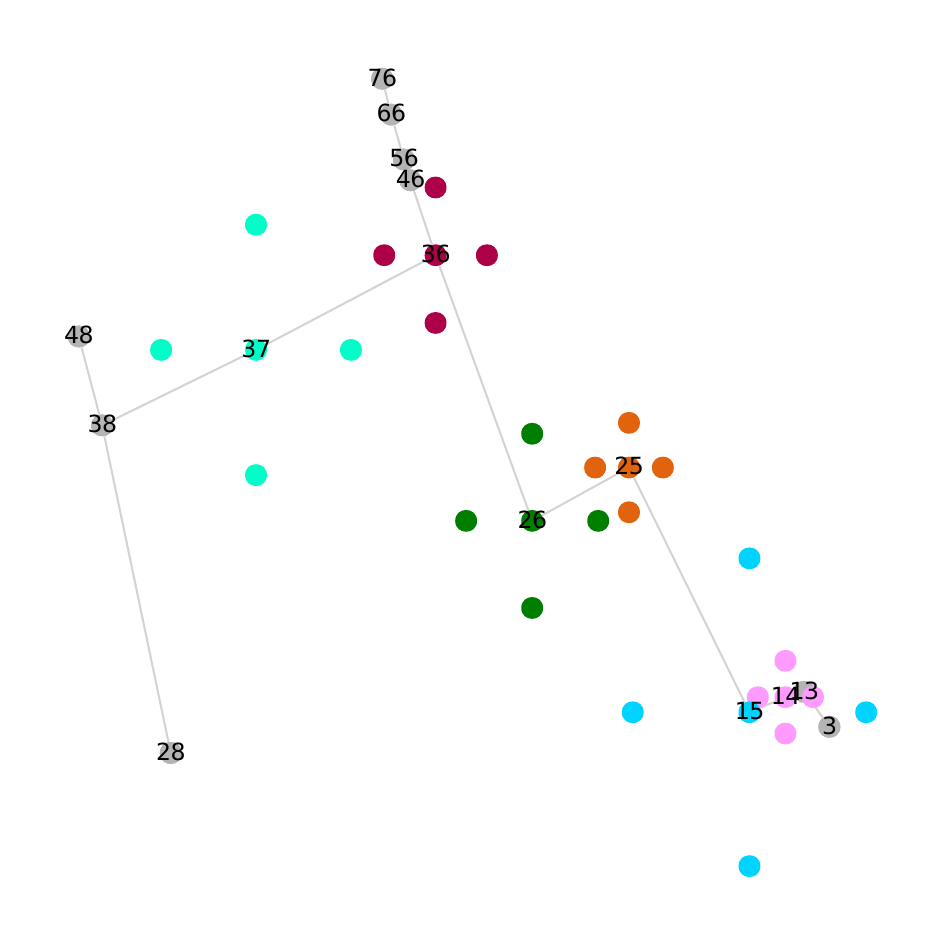}
}
\subfloat[\scriptsize \hcenter{}]{
\includegraphics[width=8cm]{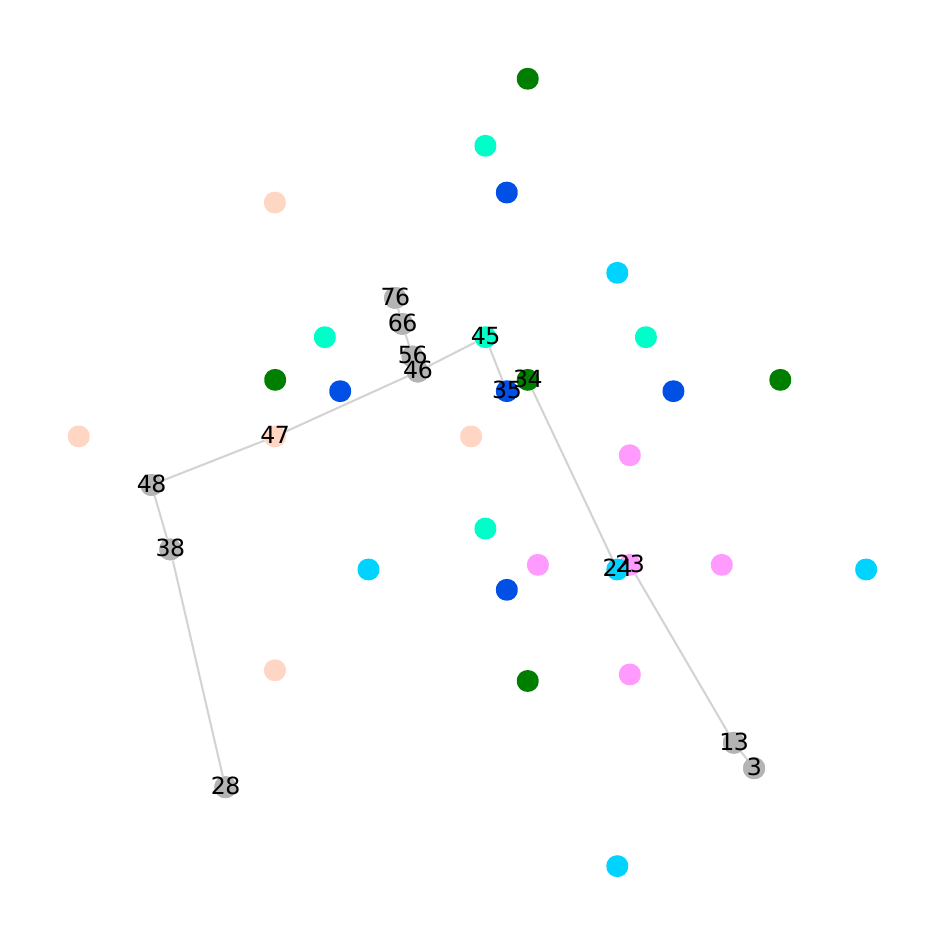}
}
\caption{\blue{Optimal solutions for instance p620 with $\Delta=0.6$ and $\nU{}=4$, for which $z(\hcenter{})/z^*=1.46$. Grey vertices are common to both solutions, while colored ones represent nominal positions and each element of $\U_i$ for the vertices that are not shared by the two solutions. We see how \hcenter{} disregards diameters, ending up with additional nodes having large uncertainty sets. More precisely, computing the average diameters of internal nodes for both solutions leads to 1974 and 3350 for \exact{} and \hcenter{}, respectively.}}
\label{fig:exres}
\end{figure}

\blue{For some insight on the poor results of \hcenter{}, one may observe that it is the only approximation that completely neglects the shape (and diameter) of the uncertainty sets. This means that the solution of center may make poor choices when selecting the non-terminal vertices included in the tree. In particular, the solution of \hcenter{} is likely to include non-terminal vertices whose uncertainty sets have much larger diameters than those selected in an optimal solution. We illustrate this by drawing an optimal solution and  the solution returned by \hcenter{} on Figure~\ref{fig:exres}. Fore more details on these aspects we refer to our companion paper~\cite{abs-2206-08187} where we study the approximation ratios of \worst{} and \hcenter{}. Among other results, we show that while \worst{} achieves a constant approximation ratio, the solution of \hcenter{} can be arbitrarily bad.}
\begin{remark}
 The optimality of the solutions returned by \cons{} that is displayed on Figure~\ref{fig:STP:costs} means that optimizing along function $$c^{\cons}(x)=\min\limits_{\mu}\max\limits_{\u\in\UU}\sum\limits_{\{i,j\}\in \EE} x_{ij}(\dist{u_i}{\mu_{ij}}+\dist{\mu_{ij}}{u_j})$$ returns the same optimal solution as optimizing along the true objective $$c(x)=\max\limits_{\u\in\UU}\sum\limits_{\{i,j\}\in \EE} x_{ij}\dist{u_i}{u_j}.$$
 It does not mean, however, that the conservative approximation (or, equivalently, the affine decision rule approximation, as discussed in Appendix~\ref{app:adr}) is exact in this case. Specifically, looking at the detailed results reveals that $c^{\cons}(x^*)>c(x^*)$ for the optimal solution $x^*$ returned by \cons.
\end{remark}

\subsection{Simple plant location}

We consider a strategic facility location problem where the exact location of the facility may be perturbed due to local political and technical considerations, while the exact position of the clients themselves is subject to uncertainty~\citep{correia2015facility}. The distances between the facilities and the clients are computed from the shortest path distance on a weighted graph that represents the underlying road network. The problem can then be modeled with the weighted graph $\GG=(\VV,\EE,l)$, the vertices of which represent the possible locations for the facilities and clients, while each edge and its weight represent the existence of a road between two vertices together with its length. The metric is induced by graph $\GG$, so $\Metric=\VV$ and $d(u,v)$ is equal to the shortest path between $u$ and $v$ for every $u,v\in V$. 

Let $I\subseteq \VV$ and $J\subseteq \VV$ represent the set of clients and possible locations for the facilities. %For every vertex $i\in I\cup J$, we consider the uncertainty set $\U_i\subseteq V$, which contains $i$ and each vertex $v\in V$ such that $d(i,v)\leq \Gamma_i$, for some given $\Gamma_i \geq 0$. 
We consider the problem of choosing $p$ facilities among $J$ and assigning every client to its closest facility so as to minimize the total assignment cost. For each $j\in J,$ let $y_j$ be a binary variable indicating whether a facility is located at $j$, and for each $i\in I, j\in J,$ let $x_{ij}\in\{0,1\}$ indicate whether client $i$ is assigned to facility $j$.
The robust problem can then be formulated as
\begin{align*}
\min &\left( \max_{\u\in \UU} \sum_{i\in I,j\in J}x_{ij}d(u_i,u_j) \right) \\
\mbox{s.t.}&\sum_{j\in J} x_{ij} = 1, \quad \forall i\in I \\
&x_{ij} \leq y_{j}, \quad \forall i\in I,j\in J \\
& \sum_{j\in J}y_j = p \\
&x,y\text{ binary}
\end{align*}

\subsubsection{Instances}

We construct the graph $\GG=(\VV,\EE,l)$ as follows. For each vertex $i$, we generate its position $u_i$ uniformly in the square $[0,1]^2$ and we select edges so that the resulting graph is planar and connected and shorter edges are more likely to appear. 
This procedure allows to mimic real transportation networks~\citep{daskin1993genrand2}. 
More precisely, we first compute a minimum cost spanning tree based on the weights $\{w_{ij}=\|u_i-u_j\|_2^{-2}\}$ to ensure the graph is connected. Then, we iteratively select $m-n+1$ additional edges following the probability distribution $Prob_{ij}=\frac{w_{ij}}{\sum_{\{i',j'\}}w_{i'j'}}$ for each $i\neq j \in \VV$ while ensuring the resulting graph is planar. The length $l_{ij}$ of each edge $\{i,j\}\in \EE$ is then given by $\|u_i-u_j\|_2$ and the the distance between every pair of vertices is given by the shortest path between them in $G$. For each $i\in\VV$, we define $\U_i$ as the $\NU$ vertices that are closest to $i$. Finally, $I$ is defined as a random subset of $\VV$ such that $\U_i\cap \U_{i'}=\emptyset$ for each $i,i'\in I$ and $J=\VV\setminus I$. Following the above procedure, we create 2 random instances for each choice of parameters $n, m, \sigma$ and $\card{I}$, leading to 486 instances.

\subsubsection{Results}

Figure~\ref{fig:SPL:times} reports the average solution times, showing that \exact{} is able to solve every instance to optimality within a few seconds, being roughly twice slower than the heuristic algorithms. The figure further underlines that $n$ is the parameter having the strongest impact on the solution time. This was expected given that larger values for $n$ imply more elements in $I$ and $J$, and therefore, larger models. The charts presented for the remaining 4 parameters do not lead to clear conclusions. 

Heuristic \cons{} is not included in the comparison because its efficiency strongly depends on the definition of sets $\Metric_{ij}$, as discussed in Section~\ref{sec:adr}. While defining $\Metric_{ij}=\Metric$ is likely to be intractable, it would lead to the tightest bounds. To obtain a good trade-off between quality and time, one should come up with ad-hoc sets $\Metric_{ij}\subset\Metric$ obtained through heuristics that would be tailored to the specific instances used. This is beyond the scope of the current paper, which aims at proposing general methods rather than ad-hoc algorithms for specific data sets.% Last but not least, \exact{} being only 2 or 3 times slower than the deterministic reformulations, there is little room left for providing yet another heuristic that would provide an even better trade-off between quality and time.

Then, following again formula~\eqref{eq:costcurves}, Figure~\ref{fig:SPL:costs} reports the cumulative distributions of the cost increase of each of the three deterministic heuristics, relatively to the cost of the exact solution. %Notice that \cons{} is not included in this comparison, because this algorithm can handle only the Euclidean distance whereas these results rely on graph-induced distances. 
The results focus only on the parameters having an impact on the resuting costs, namely $n$ and $\NU$. They illustrate that \avg{} is the best approximation, followed closely by \worst{} and \hcenter{}. They also show that \hcenter{} behaves worse for small instances and those having larger uncertainty sets.

%\jo{These results provide additional insight into the common intuition that the barycenters of the uncertainty sets should yield a good approximation of the robust problem. For instance, if the vertices are grossly distributed homogeneously on the considered region, the objective function of \hcenter{} and \worst{} are tightly connected. To illustrate this, one may consider the ideal case where the uncertainty set of each vertex $i$ is a disk with center $O_i$ and a common radius $R$. A plant $j$ and the clients $I$ allocated to $j$ correspond to a star graph $F$ centered on $j$. Now, if  the clients' centers $\{O_i\}_{i\in I}$ are symmetric with respect to $O_j$, we get that $$ \costmax(F) = \sum_{i\in I} d(O_i,O_j) + 2\card{I}\times R.$$ Consequently, the objective function of \hcenter{} and \worst{} are equal up to a constant, for all solutions where the clients are homogeneously distributed around their plant.}

Overall, these results illustrate that given the quick solution times of \exact{} due to the compact formulation presented in Section~\ref{sec:compact}, heuristic algorithms do not seem necessary for obtaining good solutions to this problem. Yet, if one wishes to reduce further the solution times, \avg{} should be preferred over \worst{} and \hcenter{}.

\begin{figure}[ht]
\centering
\subfloat[Varying $n$]{
\begin{tikzpicture}[scale=0.7]
\begin{axis}[height=5cm,width=6cm,axis x line=bottom, axis y line = left,ymin = 0,
	ylabel=\blue{{solution times (seonds)}},
	title={},
	xtick = {200,400,600},
 	every axis plot/.append style={ultra thick}]
	
	\addplot [solid, blue, mark=triangle] table[
	x=x, y=center
	]{SPL/times_n.txt};
	
	\addplot [dotted, red, mark=o, mark options=solid] table[
	x=x, y=worst
	]{SPL/times_n.txt};
	
	\addplot [dashed, teal, mark=square, mark options=solid] table[
	x=x, y=avg
	]{SPL/times_n.txt};
	
	\addplot [dotted, brown, mark=o, mark options=solid] table[
	x=x, y=exact
	]{SPL/times_n.txt};
\end{axis}
\end{tikzpicture}
}
\subfloat[Varying $|I|$]{
\begin{tikzpicture}[scale=0.7]
\begin{axis}[height=5cm,width=6cm,axis x line=bottom, axis y line = left,ymin = 0,
	ylabel=\blue{{solution times (seonds)}},
	title={},
	symbolic x coords={n/8,n/7,n/6},
	xtick={n/8,n/7,n/6},
 	every axis plot/.append style={ultra thick}]
	
	\addplot [solid, blue, mark=triangle] table[
	x=x, y=center
	]{SPL/times_I.txt};
	
	\addplot [dotted, red, mark=o, mark options=solid] table[
	x=x, y=worst
	]{SPL/times_I.txt};
	
	\addplot [dashed, teal, mark=square, mark options=solid] table[
	x=x, y=avg
	]{SPL/times_I.txt};
	
	\addplot [dotted, brown, mark=o, mark options=solid] table[
	x=x, y=exact
	]{SPL/times_I.txt};
\end{axis}
\end{tikzpicture}
}
\subfloat[Varying $p$]{
\begin{tikzpicture}[scale=0.7]
\begin{axis}[height=5cm,width=6cm,axis x line=bottom, axis y line = left,ymin = 0,
	ylabel=\blue{{solution times (seonds)}},
	title={},
	xtick= {2,4,6},
    log ticks with fixed point,
 	every axis plot/.append style={ultra thick}]
	
	\addplot [solid, blue, mark=triangle] table[
	x=x, y=center
	]{SPL/times_p.txt};
	
	\addplot [dotted, red, mark=o, mark options=solid] table[
	x=x, y=worst
	]{SPL/times_p.txt};
	
	\addplot [dashed, teal, mark=square, mark options=solid] table[
	x=x, y=avg
	]{SPL/times_p.txt};
	
	\addplot [dotted, brown, mark=o, mark options=solid] table[
	x=x, y=exact
	]{SPL/times_p.txt};
\end{axis}
\end{tikzpicture}
}
\\
\subfloat[Varying $m$]{
\begin{tikzpicture}[scale=0.7]
\begin{axis}[height=5cm,width=6cm,axis x line=bottom, axis y line = left,ymin = 0,
	ylabel=\blue{{solution times (seonds)}},
	title={},
	symbolic x  coords={n,1.5n,2n},
    xtick={n,1.5n,2n},
 	every axis plot/.append style={ultra thick}]
	
	\addplot [solid, blue, mark=triangle] table[
	x=x, y=center
	]{SPL/times_m.txt};
	
	\addplot [dotted, red, mark=o, mark options=solid] table[
	x=x, y=worst
	]{SPL/times_m.txt};
	
	\addplot [dashed, teal, mark=square, mark options=solid] table[
	x=x, y=avg
	]{SPL/times_m.txt};
	
	\addplot [dotted, brown, mark=o, mark options=solid] table[
	x=x, y=exact
	]{SPL/times_m.txt};
\end{axis}
\end{tikzpicture}
}
\subfloat[Varying $\NU$]{
\begin{tikzpicture}[scale=0.7]
\begin{axis}[height=5cm,width=6cm,axis x line=bottom, axis y line = left,ymin = 0,
	ylabel=\blue{{solution times (seonds)}},
	title={},
	%	ymode = log,
    xtick={2,3,4},
 	every axis plot/.append style={ultra thick},
 	legend entries={\hcenter, \worst, \avg, \exact},
	legend style={at={(axis cs:4.1,0.5)}, anchor=south west}]
	
	\addplot [solid, blue, mark=triangle] table[
	x=x, y=center
	]{SPL/times_nU.txt};
	
	\addplot [dotted, red, mark=o, mark options=solid] table[
	x=x, y=worst
	]{SPL/times_nU.txt};
	
	\addplot [dashed, teal, mark=square, mark options=solid] table[
	x=x, y=avg
	]{SPL/times_nU.txt};
	
	\addplot [dotted, brown, mark=o, mark options=solid] table[
	x=x, y=exact
	]{SPL/times_nU.txt};
\end{axis}
\end{tikzpicture}
}
\caption{SPL: Average solution times in seconds for each algorithm when varying one of the parameters.}
\label{fig:SPL:times}
\end{figure}
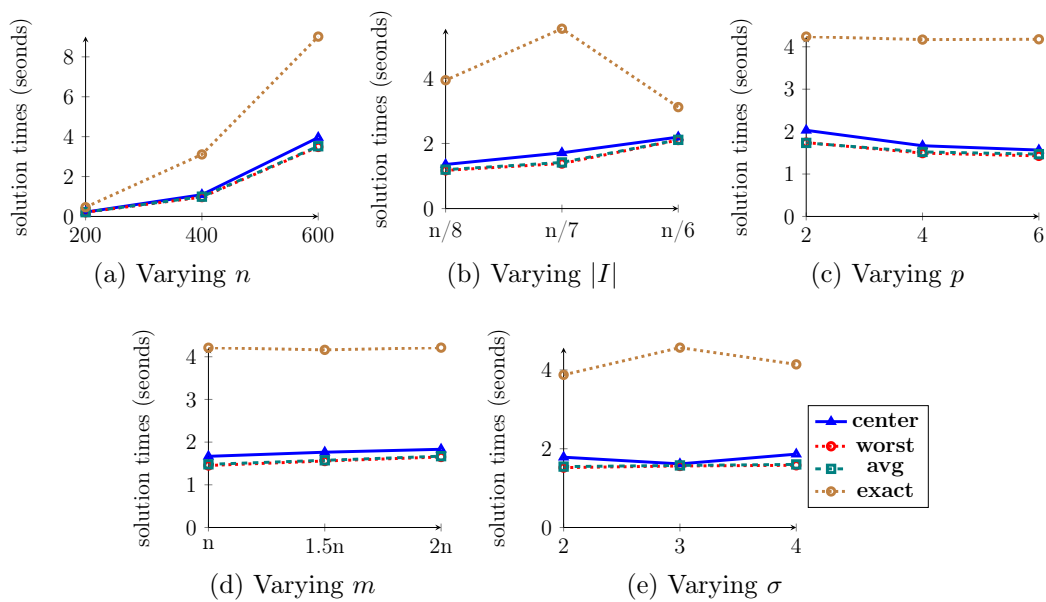
\begin{figure}[ht]
\centering
\foreach \index in {200,400,600}
{
\subfloat[$n=\index$]{
\begin{tikzpicture}[scale=0.7]
\begin{axis}[height=5cm,width=6cm,axis x line=bottom, axis y line = left,xmin = 0,ymin = 0.001, ymax = 100,
	ylabel={\blue{\% of instances}},
	title={},
    log ticks with fixed point,
 	every axis plot/.append style={ultra thick}]
	
	\addplot [solid, blue] table[
	x=x, y=solved, mark=none
	]{SPL/center_n_\index.txt};
	
	\addplot [dotted, red] table[
	x=x, y=solved, mark=none
	]{SPL/worst_n_\index.txt};
	
	\addplot [dashed, teal] table[
	x=x, y=solved, mark=none
	]{SPL/avg_n_\index.txt};
\end{axis}
\end{tikzpicture}
}
}
 \\
 \foreach \index in {2,3}
 {
 \subfloat[$\NU=\index$]{
\begin{tikzpicture}[scale=0.7]
\begin{axis}[height=5cm,width=6cm,axis x line=bottom, axis y line = left,xmin = 0,ymin = 0.001, ymax = 100,
	ylabel={\blue{\% of instances}},
	title={},
    log ticks with fixed point,
 	every axis plot/.append style={ultra thick}
 	]
	
	\addplot [solid, blue] table[
	x=x, y=solved, mark=none
	]{SPL/center_nU_\index.txt};
	
	\addplot [dotted, red] table[
	x=x, y=solved, mark=none
	]{SPL/worst_nU_\index.txt};
	
	\addplot [dashed, teal] table[
	x=x, y=solved, mark=none
	]{SPL/avg_nU_\index.txt};
\end{axis}
\end{tikzpicture}
}}
 \subfloat[$\NU=4$]{
\begin{tikzpicture}[scale=0.7]
\begin{axis}[height=5cm,width=6cm,axis x line=bottom, axis y line = left,xmin = 0,ymin = 0.001, ymax = 100,
	ylabel={\blue{\% of instances}},
	title={},
    log ticks with fixed point,
 	every axis plot/.append style={ultra thick},
    legend entries={\hcenter,\worst,\avg,\cons},
 	legend style={at={(axis cs:16,25)}, anchor=south west}
 	]
	
	\addplot [solid, blue] table[
	x=x, y=solved, mark=none
	]{SPL/center_nU_4.txt};
	
	\addplot [dotted, red] table[
	x=x, y=solved, mark=none
	]{SPL/worst_nU_4.txt};
	
	\addplot [dashed, teal] table[
	x=x, y=solved, mark=none
	]{SPL/avg_nU_4.txt};
\end{axis}
\end{tikzpicture}
}
\caption{SPL: For each heuristic algorithm, the curve plots~\eqref{eq:costcurves}, the cumulative distribution of the \% of instances for which the returned solution has an additional (relative) cost less than the value of the abscissa.}
\label{fig:SPL:costs}
\end{figure}
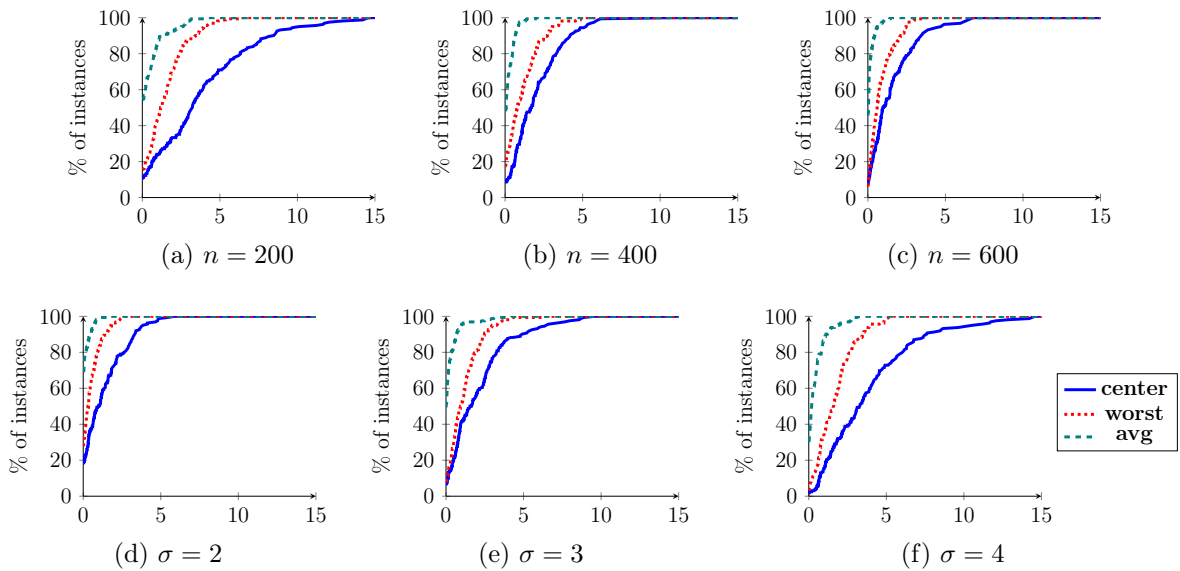

\section{Concluding remarks}

This paper has been devoted to the study of general combinatorial optimization problems defined in spatial graphs with locational uncertainty, thus encompassing applications arising in transportation and facility location, among others. 
After proving the $\NP$-hardness of these problems, we have developed an exact solution algorithm based on scenario generation. The bottleneck of this algorithm lies in the separation problem, so we have studied in depth the complexity of that problem, also proposing an integer programming formulation. We have also proposed a conservative approximation that turns out to be equivalent to the affine decision rules approximation by~\cite{zhen2021robust} in the case of Euclidean distances. We have compared these algorithms numerically to different deterministic approximations on Steiner tree and location instances inspired by the scientific literature.

Our results illustrate that the exact algorithms are fast, being able to solve in reasonable amounts of time instances of realistic sizes. They also illustrate that the deterministic reformulations based on average or worst-case distances provide very good solutions in short amounts of time, offering interesting alternatives whenever an exact solution cannot be computed in an acceptable time. 

\blue{Last, our solution algorithms and reformulations critically rely on the fact that the uncertainty sets are independent for each vertex $i\in \VV$. We believe that extending these to more general (correlated) uncertainty sets would be an interesting topic for future research.}

\bibliographystyle{informs2014} % outcomment this and next line in Case 1
\bibliography{ref}

\begin{thebibliography}{39}
\providecommand{\natexlab}[1]{#1}
\providecommand{\url}[1]{\texttt{#1}}
\providecommand{\urlprefix}{URL }

\bibitem[{Aissi et~al.(2009)Aissi, Bazgan, \protect\BIBand{}
  Vanderpooten}]{AissiBV09}
Aissi H, Bazgan C, Vanderpooten D (2009) Min-max and min-max regret versions of
  combinatorial optimization problems: {A} survey. \emph{Eur. J. Oper. Res.}
  197(2):427--438,
  \urlprefix\url{http://dx.doi.org/10.1016/j.ejor.2008.09.012}.

\bibitem[{Ayoub \protect\BIBand{} Poss(2016)}]{AyoubP16}
Ayoub J, Poss M (2016) Decomposition for adjustable robust linear optimization
  subject to uncertainty polytope. \emph{Comput. Manag. Science}
  13(2):219--239, \urlprefix\url{http://dx.doi.org/10.1007/s10287-016-0249-2}.

\bibitem[{Ben{-}Tal et~al.(2004)Ben{-}Tal, Goryashko, Guslitzer,
  \protect\BIBand{} Nemirovski}]{Ben-TalGGN04}
Ben{-}Tal A, Goryashko AP, Guslitzer E, Nemirovski A (2004) Adjustable robust
  solutions of uncertain linear programs. \emph{Math. Program.} 99(2):351--376,
  \urlprefix\url{http://dx.doi.org/10.1007/s10107-003-0454-y}.

\bibitem[{Ben-Tal \protect\BIBand{} Nemirovski(1998)}]{ben1998robust}
Ben-Tal A, Nemirovski A (1998) Robust convex optimization. \emph{Mathematics of
  Operations Research} 23(4):769--805.

\bibitem[{Bertsimas \protect\BIBand{} Dunning(2016)}]{BertsimasD16}
Bertsimas D, Dunning I (2016) Multistage robust mixed-integer optimization with
  adaptive partitions. \emph{Oper. Res.} 64(4):980--998,
  \urlprefix\url{http://dx.doi.org/10.1287/opre.2016.1515}.

\bibitem[{Bertsimas et~al.(2016)Bertsimas, Dunning, \protect\BIBand{}
  Lubin}]{BertsimasDL16}
Bertsimas D, Dunning I, Lubin M (2016) Reformulation versus cutting-planes for
  robust optimization. \emph{Comput. Manag. Science} 13(2):195--217.

\bibitem[{Bertsimas \protect\BIBand{} Sim(2003)}]{BertsimasS03}
Bertsimas D, Sim M (2003) Robust discrete optimization and network flows.
  \emph{Mathematical Programming} 98(1-3):49--71.

\bibitem[{Bezanson et~al.(2012)Bezanson, Karpinski, Shah, \protect\BIBand{}
  Edelman}]{bezanson2012julia}
Bezanson J, Karpinski S, Shah VB, Edelman A (2012) Julia: A fast dynamic
  language for technical computing. \emph{arXiv preprint arXiv:1209.5145} .

\bibitem[{Bodlaender et~al.(2013)Bodlaender, Drange, Dregi, Fomin, Lokshtanov,
  \protect\BIBand{} Pilipczuk}]{DBLP:journals/corr/abs-1304-6321}
Bodlaender HL, Drange PG, Dregi MS, Fomin FV, Lokshtanov D, Pilipczuk M (2013)
  A o(c{\^{}}k n) 5-approximation algorithm for treewidth. \emph{CoRR}
  abs/1304.6321, \urlprefix\url{http://arxiv.org/abs/1304.6321}.

\bibitem[{Bougeret et~al.(2022)Bougeret, Omer, \protect\BIBand{}
  Poss}]{abs-2206-08187}
Bougeret M, Omer J, Poss M (2022) Approximating optimization problems in graphs
  with locational uncertainty. \emph{CoRR} abs/2206.08187,
  \urlprefix\url{http://dx.doi.org/10.48550/arXiv.2206.08187}.

\bibitem[{Buchheim \protect\BIBand{} Kurtz(2018)}]{BuchheimK18}
Buchheim C, Kurtz J (2018) Robust combinatorial optimization under convex and
  discrete cost uncertainty. \emph{{EURO} J. Computational Optimization}
  6(3):211--238, \urlprefix\url{http://dx.doi.org/10.1007/s13675-018-0103-0}.

\bibitem[{Citovsky et~al.(2017)Citovsky, Mayer, \protect\BIBand{}
  Mitchell}]{citovsky2017tsp}
Citovsky G, Mayer T, Mitchell JSB (2017) {TSP With Locational Uncertainty: The
  Adversarial Model}. Aronov B, Katz MJ, eds., \emph{33rd International
  Symposium on Computational Geometry (SoCG 2017)}, volume~77 of \emph{Leibniz
  International Proceedings in Informatics (LIPIcs)}, 32:1--32:16 (Dagstuhl,
  Germany: Schloss Dagstuhl--Leibniz-Zentrum fuer Informatik), ISBN
  978-3-95977-038-5, ISSN 1868-8969,
  \urlprefix\url{http://dx.doi.org/10.4230/LIPIcs.SoCG.2017.32}.

\bibitem[{Correia \protect\BIBand{} da~Gama(2015)}]{correia2015facility}
Correia I, da~Gama FS (2015) Facility location under uncertainty.
  \emph{Location science}, 177--203 (Springer).

\bibitem[{Crescenzi(1997)}]{612321}
Crescenzi P (1997) A short guide to approximation preserving reductions.
  \emph{Proceedings of Computational Complexity. Twelfth Annual IEEE
  Conference}, 262--273,
  \urlprefix\url{http://dx.doi.org/10.1109/CCC.1997.612321}.

\bibitem[{Cygan et~al.(2015)Cygan, Fomin, Kowalik, Lokshtanov, Marx, Pilipczuk,
  Pilipczuk, \protect\BIBand{} Saurabh}]{CyganFKLMPPS15}
Cygan M, Fomin FV, Kowalik L, Lokshtanov D, Marx D, Pilipczuk M, Pilipczuk M,
  Saurabh S (2015) \emph{Parameterized Algorithms} (Springer),
  \urlprefix\url{http://dx.doi.org/10.1007/978-3-319-21275-3}.

\bibitem[{Daskin(1993)}]{daskin1993genrand2}
Daskin M (1993) Genrand2: A random network generator. \emph{Department of
  Industrial Engineering and Management Sciences, Northwestern University,
  Evanston, IL, USA} .

\bibitem[{de~Ruiter \protect\BIBand{} Ben-Tal(2017)}]{de2017tractable}
de~Ruiter F, Ben-Tal A (2017) Tractable nonlinear decision rules for robust
  optimization. Technical report, Working paper.

\bibitem[{Downey \protect\BIBand{} Fellows(2013)}]{DF13}
Downey RG, Fellows MR (2013) \emph{Fundamentals of Parameterized Complexity}.
  Texts in Computer Science (Springer),
  \urlprefix\url{http://dx.doi.org/10.1007/978-1-4471-5559-1}.

\bibitem[{Dunning et~al.(2017)Dunning, Huchette, \protect\BIBand{}
  Lubin}]{DunningHuchetteLubin2017}
Dunning I, Huchette J, Lubin M (2017) Jump: A modeling language for
  mathematical optimization. \emph{SIAM Review} 59(2):295--320,
  \urlprefix\url{http://dx.doi.org/10.1137/15M1020575}.

\bibitem[{Fellows et~al.(2011)Fellows, Fomin, Lokshtanov, Rosamond, Saurabh,
  Szeider, \protect\BIBand{} Thomassen}]{fellows2011complexity}
Fellows MR, Fomin FV, Lokshtanov D, Rosamond F, Saurabh S, Szeider S, Thomassen
  C (2011) On the complexity of some colorful problems parameterized by
  treewidth. \emph{Information and Computation} 209(2):143--153.

\bibitem[{Fischetti \protect\BIBand{} Monaci(2012)}]{FischettiM12}
Fischetti M, Monaci M (2012) Cutting plane versus compact formulations for
  uncertain (integer) linear programs. \emph{Math. Program. Comput.}
  4(3):239--273.

\bibitem[{Gutiérrez-Jarpa et~al.(2013)Gutiérrez-Jarpa, Obreque, Laporte,
  \protect\BIBand{} Marianov}]{GUTIERREZJARPA20133000}
Gutiérrez-Jarpa G, Obreque C, Laporte G, Marianov V (2013) Rapid transit
  network design for optimal cost and origin–destination demand capture.
  \emph{Computers \& Operations Research} 40(12):3000--3009, ISSN 0305-0548,
  \urlprefix\url{http://dx.doi.org/https://doi.org/10.1016/j.cor.2013.06.013}.

\bibitem[{Hanasusanto et~al.(2015)Hanasusanto, Kuhn, \protect\BIBand{}
  Wiesemann}]{hanasusanto2015k}
Hanasusanto GA, Kuhn D, Wiesemann W (2015) K-adaptability in two-stage robust
  binary programming. \emph{Operations Research} 63(4):877--891.

\bibitem[{Kasperski \protect\BIBand{}
  Zieli{\'n}ski(2009)}]{kasperski2009approximability}
Kasperski A, Zieli{\'n}ski P (2009) On the approximability of minmax (regret)
  network optimization problems. \emph{Information Processing Letters}
  109(5):262--266.

\bibitem[{Kasperski \protect\BIBand{}
  Zieli{\'n}ski(2016)}]{kasperski2016robust}
Kasperski A, Zieli{\'n}ski P (2016) Robust discrete optimization under discrete
  and interval uncertainty: A survey. \emph{Robustness analysis in decision
  aiding, optimization, and analytics}, 113--143 (Springer).

\bibitem[{Kloks(1994)}]{Klo94}
Kloks T (1994) \emph{Treewidth, Computations and Approximations}, volume 842 of
  \emph{Lecture Notes in Computer Science} (Springer-Verlag),
  \urlprefix\url{http://dx.doi.org/10.1007/BFb0045375}.

\bibitem[{Kouvelis \protect\BIBand{} Yu(2013)}]{kouvelis2013robust}
Kouvelis P, Yu G (2013) \emph{Robust discrete optimization and its
  applications}, volume~14 (Springer Science \& Business Media).

\bibitem[{Magnanti \protect\BIBand{} Wong(1984)}]{magnanti1984network}
Magnanti TL, Wong RT (1984) Network design and transportation planning: Models
  and algorithms. \emph{Transportation science} 18(1):1--55.

\bibitem[{Marín \protect\BIBand{} Pelegrín(2019)}]{Marin2019}
Marín A, Pelegrín M (2019) P-{{Median}} problems. \emph{Location Science},
  25--50 ({Springer International Publishing}), ISBN 978-3-030-32177-2,
  \urlprefix\url{http://dx.doi.org/10.1007/978-3-030-32177-2_2}.

\bibitem[{Melkote \protect\BIBand{} Daskin(2001)}]{MELKOTE2001515}
Melkote S, Daskin MS (2001) An integrated model of facility location and
  transportation network design. \emph{Transportation Research Part A: Policy
  and Practice} 35(6):515--538, ISSN 0965-8564,
  \urlprefix\url{http://dx.doi.org/https://doi.org/10.1016/S0965-8564(00)00005-7}.

\bibitem[{Naoum{-}Sawaya \protect\BIBand{} Buchheim(2016)}]{Naoum-SawayaB16}
Naoum{-}Sawaya J, Buchheim C (2016) Robust critical node selection by benders
  decomposition. \emph{{INFORMS} J. Comput.} 28(1):162--174,
  \urlprefix\url{http://dx.doi.org/10.1287/ijoc.2015.0671}.

\bibitem[{Postek \protect\BIBand{} Den~Hertog(2016)}]{postek2016multistage}
Postek K, Den~Hertog D (2016) Multistage adjustable robust mixed-integer
  optimization via iterative splitting of the uncertainty set. \emph{INFORMS
  Journal on Computing} 28(3):553--574.

\bibitem[{Roos et~al.(2018)Roos, Den~Hertog, Ben-Tal, de~Ruiter,
  \protect\BIBand{} Zhen}]{roos2018approximation}
Roos E, Den~Hertog D, Ben-Tal A, de~Ruiter FJ, Zhen J (2018) Approximation of
  hard uncertain convex inequalities. \emph{Optimization Online URL
  http://www.optimization-online.org/DB\_HTML/2018/06/6679. html} .

\bibitem[{Shang et~al.(2003)Shang, Ruml, Zhang, \protect\BIBand{}
  Fromherz}]{shang2003localization}
Shang Y, Ruml W, Zhang Y, Fromherz MP (2003) Localization from mere
  connectivity. \emph{Proceedings of the 4th ACM international symposium on
  Mobile ad hoc networking \& computing}, 201--212.

\bibitem[{Subramanyam et~al.(2019)Subramanyam, Gounaris, \protect\BIBand{}
  Wiesemann}]{subramanyam2019k}
Subramanyam A, Gounaris CE, Wiesemann W (2019) K-adaptability in two-stage
  mixed-integer robust optimization. \emph{Mathematical Programming
  Computation} 1--32.

\bibitem[{Yaman et~al.(2001)Yaman, Karasan, \protect\BIBand{}
  Pinar}]{YamanKP01}
Yaman H, Karasan OE, Pinar M{\c{C}} (2001) The robust spanning tree problem
  with interval data. \emph{Oper. Res. Lett.} 29(1):31--40,
  \urlprefix\url{http://dx.doi.org/10.1016/S0167-6377(01)00078-5}.

\bibitem[{Zeng \protect\BIBand{} Zhao(2013)}]{zeng2013solving}
Zeng B, Zhao L (2013) Solving two-stage robust optimization problems using a
  column-and-constraint generation method. \emph{Operations Research Letters}
  41(5):457--461.

\bibitem[{Zhen et~al.(2021)Zhen, de~Ruiter, Roos, \protect\BIBand{} den
  Hertog}]{zhen2021robust}
Zhen J, de~Ruiter FJ, Roos E, den Hertog D (2021) Robust optimization for
  models with uncertain second-order cone and semidefinite programming
  constraints. \emph{INFORMS Journal on Computing} .

\bibitem[{Zhen et~al.(2018)Zhen, Hertog, \protect\BIBand{} Sim}]{ZhenHS18}
Zhen J, Hertog DD, Sim M (2018) {Adjustable Robust Optimization via
  Fourier-Motzkin Elimination}. \emph{Operations Research} 66(4):1086--1100,
  \urlprefix\url{http://dx.doi.org/10.1287/opre.2017.1714}.

\end{thebibliography}

\begin{appendix}

%%%%%%%%%%%%%%%%%%%%%%%%%%%%%%%%%%%%%%%%%%%%%%%%%%%%%%%%%%%%%%%%%%%%%%%%%%%%%%%%%%%%%%%%%%%%%%%%%%%%%%%%%%%
\section{Definitions related to parameterized complexity and treewidth}
\label{app:fptandtw}
%%%%%%%%%%%%%%%%%%%%%%%%%%%%%%%%%%%%%%%%%%%%%%%%%%%%%%%%%%%%%%%%%%%%%%%%%%%%%%%%%%%%%%%%%%%%%%%%%%%%%%%%%%%

\subsection{Tree decompositions and treewidth.}
\label{app:tw}

A \emph{tree decomposition} of a graph $G=(V,E)$ is a pair ${\cal D}=(T,{\cal B})$, where $T$ is a tree
and ${\cal B}=\{X^{w}\mid w\in V[T]\}$ is a collection of subsets of $V$, called \emph{bags},
% indexed by the vertices of $T$,
such that:
\begin{itemize}
\item $\bigcup_{w \in V[T]} X^w = V$,
\item for every edge $\{i,j\} \in E$, there is a $w \in V[T]$ such that $\{i, j\} \subseteq X^w$, and
\item for every $\{x,y,z\} \subseteq V[T]$ such that $z$ lies on the unique path between $x$ and $y$ in $T$,  $X^x \cap X^y \subseteq X^z$.
\end{itemize}
We call the vertices of $T$ {\em vertices} of ${\cal D}$ and the sets in ${\cal B}$ {\em bags} of ${\cal D}$. The \emph{width} of a  tree decomposition ${\cal D}=(T,{\cal B})$ is $\max_{w \in V[T]} |X^w| - 1$.
The \emph{treewidth} of a graph $G$, denoted by $tw(G)$, is the smallest integer $t$ such that there exists a tree decomposition of $G$ of width at most $t$.
%For each $t \in V[T]$, we denote by $E_w$ the set $E(G[X_w])$ \ig{do we use it?}.
Let us now recall the definition of a \emph{nice tree decomposition}, which will make the presentation of the algorithm used to proof Theorem~\ref{thm:tw} much simpler.

Let ${\cal D}=(T,{\cal B})$
be a rooted tree decomposition of $G$ (meaning that $T$ has a special vertex $r$ called the \emph{root}).
As $T$ is rooted, we naturally define an ancestor relation among bags, and say that $X^{w'}$ is a \emph{descendant} of $X^w$ if the vertex set of the unique simple path in $T$ from $r$ to $w'$ contains
$w$. In particular, every vertex $w$ is a descendant of itself.
For every $w \in V[T]$, we define $G^w=G[\bigcup\{X^{w'}\mid X^{w'}\text{ is a descendant of }X^w\text{ in }T\}]$.

Such a rooted decomposition is called a \emph{nice tree decomposition} of $G$ if the following conditions hold:
\begin{itemize}

\item $X^{r} = \emptyset$.
\item Every vertex of $T$ has at most two children in $T$.
\item For every leaf $\ell \in V[T]$, $X^{\ell} = \emptyset$. Each such vertex $\ell$ is called a {\em leaf vertex}.
\item If $w \in V[T]$ has exactly one child $w'$, then either
  \begin{itemize}
  \item $X^w = X^{w'}\cup \{i\}$ for some $i \not \in X^{w'}$.
    Each such vertex is called an \emph{introduce vertex}.%  and the vertex $v_{\rm in}$ is the {\em introduced vertex} of $X_{w}$,
    % \item $X_w = X_{w'}$ and $G_{w}=(G_{w'},E(G_{w'})\cup\{e_{\rm insert}\})$ where $e_{\rm insert}$ is an edge of $G$ with endpoints in  $X_{w}$.
    %   The vertex $t$ is called \emph{introduce edge} vertex  and the edge $e_{\rm insert}$ is the {\em insertion edge} of $X_{w}$, or
  \item $X^w = X^{w'} \setminus \{i\}$ for some $i \in X^{w'}$.
    Each such vertex is called a \emph{forget vertex}. % vertex and $v_{\rm out}$ is the {\em forget vertex} of $X_{w}$.
  \end{itemize}
\item If $w \in V[T]$ has exactly two children $w_L$ and $w_R$, then $X^{w} = X^{w_L} = X^{w_R}$. %, $E(G_{w_1})\cap E(G_{w_2})=E(G[X_w])$, and $G_w = (V(G_{w_1})\cup V(G_{w_2}),E(G_{w_1})\cup E(G_{w_2}))$.
 Each such vertex $w$ is called a \emph{join vertex}.
\end{itemize}
We recall that one of the key property of such a nice decomposition is that for any $w \in V[T]$, $X^w$ is a separator of $G$.
This implies in particular that, in a join vertex, there is no edge $\{i,j\}\in G^w$ such that $i \in V[G^{w_L}] \setminus X^w$ and $j \in V[G^{w_R}] \setminus X^w$.

Given a tree decomposition of a graph $G$ of width $t$ and $x$ vertices, it is possible to transform it in polynomial time into a \emph{nice} one of width $t$ and $xt$ vertices~\citep{Klo94}.
Moreover, it is possible (\cite{DBLP:journals/corr/abs-1304-6321}) to compute a tree decomposition of width $tw' = \grandO(tw(G))$ and $\grandO(n)$ vertices in time $\grandO(c^{tw(G)}n)$, where $n = |V|$.
By using these two results, we can compute in time $\grandO(c^{tw(G)}n)$ a nice tree decomposition of width $\grandO(tw(G))$ with $\grandO(tw(G)n)$ vertices.

\subsection{Parameterized complexity}
\label{app:fpt}

We refer the reader to~\cite{DF13,CyganFKLMPPS15} for basic background on parameterized complexity, and we recall here only some basic definitions. A \emph{parameterized problem} is a language $L \subseteq \Sigma^* \times \mathbb{N}$, where $\Sigma$ is some fixed alphabet.  For an instance $I=(x,k) \in \Sigma^* \times \mathbb{N}$, $k$ is called the \emph{parameter}.
Given a classical (non-parameterized) decision problem $L_{c} \subseteq \Sigma^*$ and a function $\kappa: \Sigma^* \rightarrow \mathbb{N}$, we denote by
$L_{c}/\kappa = \{(x,\kappa(x)\} \mid x \in L_{c}\}$ the associated parameterized problem.

A parameterized problem $L$ is \emph{fixed-parameter tractable} (\FPT) if there exists an algorithm $\mathcal{A}$, a computable function $f$, and a constant $c$ such that given an instance $I=(x,k)$, $\mathcal{A}$   (called an \FPT \emph{algorithm}) correctly decides whether $I \in L$ in time bounded by $f(k) \cdot |I|^c$.
For instance, the \textsc{Vertex Cover} problem parameterized by the size of the solution is \FPT.

Within parameterized problems, the $\cal{W}$-hierarchy may be seen as the parameterized equivalent to the class $\NP$ of classical decision problems. Without entering into details (see~\cite{DF13,CyganFKLMPPS15} for the formal definitions), a parameterized problem being \WONEH can be seen as a strong evidence that this problem is \textsl{not} \FPT.
The canonical example of \WONEH problem is \textsc{Independent Set}  parameterized by the size of the solution.

The most common way to transfer \WONEHness is via parameterized reductions.
A \emph{parameterized reduction} from a parameterized problem $L_1$ to a parameterized problem $L_2$ is an algorithm that, given an instance $(x,k)$ of $L_1$, outputs an instance $(x',k')$ of $L_2$
such that
\begin{itemize}
\item $(x,k)$ is a yes-instance of $L_1$ if and only if $(x',k')$ is a yes-instance of $L_2$,
\item $k' \le g(k)$ for some computable function $g$, and
\item the running time is bounded by $f(k)\cdot|x|^{\grandO(1)}$ for some computable function $f$.
\end{itemize}
If $L_1$ is \WONEH and there is a parameterized reduction from $L_1$ to $L_2$, then $L_2$ is \WONEH as well.

%For the sake of simplicity of the (already quite heavy) notation used in the dynamic programming algorithm of Section~\ref{sec:tw}, we will drop the vertices of $V[T]$ from the notation of bags defined above. Therefore,
%in the case of an introduce or forget vertex, the bag $X^w$ and its child $X^{w'}$ will be denoted $X$ and $X^C$, respectively,
%and in the case of a join vertex, the bag $X^w$ and its children $X^{w_L}$ and $X^{w_R}$ will be denoted $X,X^L,$ and $X^R$ respectively.

%%%%%%%%%%%%%%%%%%%%%%%%%%%%%%%%%%%%%%%%%%%%%%%%%%%%%%%%%%%%%%%%%%%%%%%%%%%%%%%%%%%%%%%%%%%%%%%%%%%%%%%%%%%
\section{Computing the objective function on small treewidth graphs}
\label{app:tw2}

Throughout this section, we consider the graph $G=(V,E)$ and denote by $\u_{|X}$ the vector $\u$ restricted to components $u_i$ such that $i \in X$, for any $X \subseteq V$.

\subsection{Definition of the auxiliary problem}
In this section we consider that we are given a fixed input of \ref{eq:evalc}, and a nice tree decomposition ${\cal D}=(T,{\cal B})$ of $G$.
Given $w \in V[T]$, we denote $\UU^w = \times_{i\in V[G^w]}\U_i$.
Let us define the following maximization problem \ref{eq:CO}. 
An input of \ref{eq:CO} is a pair $(w,f)$ where $w \in V[T]$, and $f$ is a function from $X^w$ to $\Metric$ such that for any $i \in X^w$, $f(i) \in \U_i$.
An output is a vector $\u \in \UU^w$ such that for any $i \in X^w, u_i = f(i)$, which we denote by $\u \tchak (w,f)$.
The objective is to maximize $c(\u,G^w)$. We denote by $\OPT(w,f)$ the optimal value for instance $(w,f)$.
As usual in DP algorithms, to simplify the presentation we will define an algorithm $A$ that given an input $(w,f)$ only computes the value $\OPT(w,f)$.
This algorithm could be easily modified to get an associated optimal solution.

\subsection{Join case}
Let $w$ be a join vertex with children $w^L$ and $w^R$.
Given two vectors $\u^L \in \UU^{w^L}$ and $\u^R \in \UU^{w^R}$, such that for any $i \in X^w, \u^L_i = \u^R_i$,
we define $\u = \u^L \diamond \u^R$ by $u_i = \u^L_i$ for any $i \in V[G^{w^L}]$, and $u_i = \u^R_i$ for any $i \in V[G^{w^R}]$.
Observe that $\u$ is well defined as for $i \in X^w$, $\u^L_i = \u^R_i$.

\begin{lemma}\label{lemma:join1}
  Let $(w,f)$ be an input of \ref{eq:CO} such that $w$ is a join vertex with children $w^L$ and $w^R$.
  For any $\u \in \UU^w$, $\u \tchak (w,f)$ if and only if there exists $\u^L$, $\u^R$ such that the following conditions hold:
  \begin{itemize}
  \item $\u^L \tchak (w^L,f)$
  \item $\u^R \tchak (w^R,f)$
  \item $\u = \u^L \diamond \u^R$
  \end{itemize}
\end{lemma}
\begin{proof} 
\blue{We obtain the $\Rightarrow$ direction} by defining $\u^L=\u_{|V[G^{w_L}]}$ (resp. $\u^R=\u_{|V[G^{w_R}]}$).
In the $\Leftarrow$ direction, observe that $\u^L \diamond \u^R$ is well defined as for any $i \in X^w, \u^L_i = \u^R_i = f(i)$, and
$\u \tchak (w,f)$ is also immediate.
\end{proof}

\begin{lemma}\label{lemma:join2}
   Let $(w,f)$ be an input of \ref{eq:CO} such that $w$ is a join vertex with children $w^L$ and $w^R$. Then, 
  $\OPT(w,f) = \OPT(w^L,f)+\OPT(w^R,f)-d^{(w,f)}$, where $d^{(w,f)} = \sum_{i,j \in X^w, \{i,j\} \in E[G]} d(f(i),f(j))$.
\end{lemma}
\begin{proof} 
  Let us start with the $\le$ inequality.
  Let $\u$ such that $c(\u,G^w)=\OPT(w,f)$. Let $\u^L$ and $\u^R$ as defined by Lemma~\ref{lemma:join1}.
  Observe that $c(\u,G^w) = c(\u,G^{w_L})+ c(\u,G^{w_R})-d^{(w,f)}$ as edges inside $X^w$ are counted twice in the first two terms.
  We have $c(\u,G^{w_L}) = c(\u^L,G^{w_L})$, and $c(\u^L,G^{w_L}) \le OPT(w^L,f)$ as $\u^L \tchak (w^L,f)$, and same properties hold for the right side.
  This implies $\OPT(w,f) \le \OPT(w^L,f)+\OPT(w^R,f)-d^{(w,f)}$.

  Let us now turn to the other inequality.
  Let $\u^L$ such that $c(\u^L,G^{w_L})=\OPT(w^L,f)$, $\u^R$ such that $c(\u^R,G^{w_R})=\OPT(w^R,f)$, and $\u = \u^L \diamond \u^R$.
  According to Lemma~\ref{lemma:join1}, $\u \tchak (w,f)$, and again $c(\u,G^w) = c(\u^L,G^{w_L})+ c(\u^R,G^{w_R})-d^{(w,f)}$, implying the desired inequality. 
\end{proof}

We are now ready to define the DP algorithm $A$ in the join case.
Given an input $(w,f)$ of \ref{eq:CO} such that $w$ is a join vertex with children $w^L$ and $w^R$, $A(w,f)$ returns $A(w^L,f)+A(w^R,f)-d^{(w,f)}$.
It \blue{follows} from induction and using Lemma~\ref{lemma:join2} that $A(w,f)=\OPT(w,f)$.

%% Let $U_i$ be a polytope with $s$ extreme points $\ext(U_i)=\{u_i^1,\ldots,u_i^{s_i}\}$ for each $i\in\VV$ and let $x$ describe a tree in $G$. Then, the separation problem~\eqref{eq:separation} can be solved in polynomial time. We describe the algorithm for the path $1\leftrightarrow 2 \leftrightarrow\cdots \leftrightarrow n$ for simplicity. Let $C(i,v_i)$ be the value obtained from vertex $1$, given that $u_i=v_i$. We have the following recurrence relation, running in $O(s^2n)$:
%% $$
%% C(i,v_i) = \left\{
%% \begin{array}{ll}
%% \max\limits_{v_{i+1} \in \ext(U_{i+1})} h_{i(i+1)}(d_{i(i+1)}(v_i,v_{i+1})) + C(i+1,v_{i+1}),& i \leq n-1\\
%% 0, & i = n
%% \end{array}
%% \right.
%% $$

\subsection{Introduce case}
Given any input $(w,f)$ of \ref{eq:CO} and $X \subseteq X^w$, we denote by $f_{|X}$ function $f$ restricted to $X$. The following two lemmas are \blue{easily verified}.
\begin{lemma}\label{lemma:intro1}
  Let $(w,f)$ be an input of \ref{eq:CO} such that $w$ is an introduce vertex with children $w'$.
  Let $i$ be such that $X^w = X^{w'} \cup \{i\}$.
  For any $\u \in \UU^w$, $\u \tchak (w,f)$ if and only if the following conditions hold:
  \begin{itemize}
  \item $\u_{|V[G^{w'}]} \tchak (w',f_{|X^{w'}})$
  \item $u_i = f(i)$
  \end{itemize}
\end{lemma}

\begin{lemma}\label{lemma:intro2}
  Let $(w,f)$ be an input of \ref{eq:CO} such that $w$ is an introduce vertex with children $w'$.
  Let $i$ be such that $X^w = X^{w'} \cup \{i\}$. Then, 
  $\OPT(w,f) = \OPT(w',f_{|X^{w'}})+d^{(i,w,f)}$, where $d^{(i,w,f)} = \sum_{j \in X^w, \{i,j\} \in E[G]} d(f(i),f(j))$.
\end{lemma}

We are now ready to define the DP algorithm $A$ in the introduce case.
Given an input $(w,f)$ of \ref{eq:CO} such that $w$ is an introduce vertex with children $w'$, where $X^w = X^{w'} \cup \{i\}$,
$A(w,f)$ returns $A(w',f_{|X^{w'}})+d^{(i,w,f)}$.
Using Lemma~\ref{lemma:intro2} \blue{and induction, we obtain} that $A(w,f)=\OPT(w,f)$.

\subsection{Forget case}
Let $(w,f)$ be an input of \ref{eq:CO} such that $w$ is a forget vertex with children $w'$.
Let $i$ such that $X^{w'} = X^{w} \cup \{i\}$.
For any $x \in \Metric$, we denote $f^{(i,x)}$ the function from $X^{w'}$ to $\Metric$ such that $f^{(i,x)}(j)=f(j)$ for any $j \neq i$,
and $f^{(i,x)}(i)=x$.
\blue{We obtain the following lemma}.
\begin{lemma}\label{lemma:forget1}
  Let $(w,f)$ be an input of \ref{eq:CO} such that $w$ is a forget vertex with children $w'$.
  Let $i$ be such that $X^{w'} = X^{w} \cup \{i\}$.
  For any $\u \in \UU^w$, $\u \tchak (w,f)$ if and only if $\u \tchak (w',f^{(i,u_i)})$.
\end{lemma}

\begin{lemma}\label{lemma:forget2}
  Let $(w,f)$ be an input of \ref{eq:CO} such that $w$ is a forget vertex with children $w'$.
  Let $i$ be such that $X^{w'} = X^{w} \cup \{i\}$. Then, 
  $\OPT(w,f) = \max_{x \in \U_i}\OPT(w',f^{(i,x)})$.
\end{lemma}

\begin{proof} 
  Observe first that $G^w = G^{w'}$.
  Let us start with the $\le$ inequality.
  Let $\u$ such that $c(\u,G^w)=\OPT(w,f)$. Notice that $c(\u,G^w)=c(\u,G^{w'})$.
  By Lemma~\ref{lemma:forget1}, $\u \tchak (w',f^{(i,u_i)})$, implying $c(\u,G^{w'}) \le \OPT(w',f^{(i,u_i)}) \le \max_{x \in \U_i}\OPT(w',f^{(i,x)})$.
  
  Let us now turn to the other inequality.
  Let $x^* \in \U_i$ maximizing the right side.
  Let $\u$ such that $c(\u,G^{w'})=\OPT(w',f^{(i,x^*)})$. Notice that as $\u \tchak (w',f^{(i,x^*)}), u_i = x^*$, and thus $\u \tchak (w',f^{(i,u_i}))$.
  According to Lemma~\ref{lemma:forget1}, $\u \tchak (w,f)$, implying that $c(\u,G^{w'}) = c(\u,G^{w}) \le \OPT(w,f)$.
\end{proof}

We are now ready to define the DP algorithm $A$ in the forget case.
Given an input $(w,f)$ of \ref{eq:CO} such that $w$ is a forget vertex with children $w'$, where $X^{w'} = X^{w} \cup \{i\}$,
$A(w,f)$ returns $\max_{x \in \U_i}A(w',f^{(i,x)})$.
It \blue{follows} by induction and using Lemma~\ref{lemma:forget2} that $A(w,f)=\OPT(w,f)$.

\subsection{Putting pieces together}

\begin{theorem}%\label{thm:tw}
  \EVALC/$tw+\NU$ is \FPT. More precisely, we can compute an optimal solution of \ref{eq:evalc} in time $\grandO(n tw \NU^{\grandO(tw)})$, where $n = |V|$, $tw = tw(G)$,
  and $\NU= \max_{i \in V}\nU{i}$.
\end{theorem}
\begin{proof} 
  Given an input $(\Metric,d,G,\UU)$ of \ref{eq:evalc}, we start (see Appendix~\ref{app:tw}) by computing in time $\grandO(c^{tw(G)}n)$ a nice tree decomposition of width $\grandO(tw(G))$
  with $N = \grandO(n tw(G))$ vertices.
  Remember that this nice tree decomposition is rooted on a vertex $r$ such that $X^r=\emptyset$.
  Then, we output $A(r,\emptyset)$. Notice that as $X^r = \emptyset$, the second parameter (the function from $X^r$ to $\Metric$) is defined nowhere and denoted $\emptyset$.
   As $A$ solves \ref{eq:CO} optimaly, we have $A(r,\emptyset)=\OPT(r,\emptyset)$. Moreover, as $G^r = G$, we have $\OPT(r,\emptyset)=c(G)$.

   Let us now consider the running time of $A$. Given a tree decomposition with $N$ vertices (in the tree of bags) and of width $t$,
   the size of the DP table is $\grandO(N \NU^t)$,  the time to compute one entry is dominated by the forget case where the branching is in $\grandO(\NU)$,
   implying a running time in $\grandO(N \NU^{t+1})$. Pluging the corresponding values, we get the claimed running time.
\end{proof}

\section{Compact formulation for the Steiner Tree Problem}
\label{app:compactSTP}

We extend next the construction from Section~\ref{sec:compact} to trees, albeit this involves logical constraints. We consider more particularly the case of the Steiner Tree Problem where a set of terminals $T\subseteq V$ is given and any feasible solution is a Steiner tree connecting the terminals of $T$. We further assume that $r$ is a given arbitrary root in $T$ and that any $x\in\X$ describes an directed tree from $r$ to the set of terminals $T\setminus \{r\}$. In particular, this involves that the edges are directed, so variables $x_{ij}$ and $x_{ji}$ now denote the two directed edges $(i,j)$ and $(j,i)$ obtained from $\{i,j\}$, leading to the directed set of edges $\EE^{bidir}$. Similarly, we introduce the incoming and outgoing stars of $i$ as $\delta^-(i)=\set{j}{(j,i)\in \EE^{bidir}}$ and $\delta^+(i)=\set{j}{(i,j)\in \EE^{bidir}}$, \blue{respectively}.

Then, extending the optimization variable $z$ to any node in $V$, we obtain the following formulation
\begin{align}
\min\quad & \omega \nonumber\\
\mbox{s.t.}\quad & \omega \geq z_{\root}^k, \quad \forall k\in [\nU{\root}]\\
& z_{i}^k \geq \sum_{j \in \delta^+(i)} x_{ij} \max_{\ell\in[\nU{j}]}\left(\dist{u_i^k}{u_j^\ell}+z_{j}^\ell\right), \quad \forall i \in V, k\in [\nU{j}]  \label{eq:v1}\\
& x \in \X, z\geq 0.
\end{align}
To linearize the maxima in the right-hand-side of~\eqref{eq:v1}, we introduce variables $Z_{ij}^{k}$ such that
$$
Z_{ij}^{k} \geq \dist{u_i^k}{u_j^\ell}+z_{j}^\ell, \quad \forall \ell\in[\nU{j}]
$$
and replace~\eqref{eq:v1} with
\begin{equation}
\label{eq:zxZ}
z_{i}^k \geq \sum_{j \in \delta^+(i)} x_{ij} Z_{ij}^{k}, \quad \forall i \in V\setminus T_0, k,\ell\in [\nU{j}].
\end{equation}
The right-hand-side of constraints~\eqref{eq:zxZ} can be further linearized with the help of additional variables $X_{ij}^{k}$ and logical constraints
$$
x_{ij}=1\implies X_{ij}^{k}\geq Z_{ij}^{k}.
$$

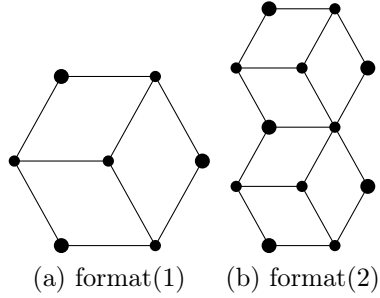
\begin{figure}
\centering
\subfloat[format(1)]{
\begin{tikzpicture}[scale=0.5]
\tkzDefPoint(2.75,0.25){u1}
\tkzDefPoint(5.25,0.25){u2}
\tkzDefPoint(1.5,2.5){u3}
\tkzDefPoint(4,2.5){u4}
\tkzDefPoint(6.5,2.5){u5}
\tkzDefPoint(2.75,4.75){u6}
\tkzDefPoint(5.25,4.75){u7}
\foreach \n in {u1,u5,u6}
  \node at (\n)[circle,fill,inner sep=2pt]{};
\foreach \n in {u2,u3,u4,u7}
  \node at (\n)[circle,fill,inner sep=1.5pt]{};
\foreach \v/\w in {u1/u2, u1/u3, u2/u4, u3/u4, u2/u5, u3/u6, u4/u7, u5/u7, u6/u7}
    \draw[-] (\v) to (\w);
\end{tikzpicture}
}
\subfloat[format(2)]{
\begin{tikzpicture}[scale=0.35]
\tkzDefPoint(2.75,0.25){u1}
\tkzDefPoint(5.25,0.25){u2}
\tkzDefPoint(1.5,2.5){u3}
\tkzDefPoint(4,2.5){u4}
\tkzDefPoint(6.5,2.5){u5}
\tkzDefPoint(2.75,4.75){u6}
\tkzDefPoint(5.25,4.75){u7}

\tkzDefPoint(1.5,7){u8}
\tkzDefPoint(4,7){u9}
\tkzDefPoint(6.5,7){u10}
\tkzDefPoint(2.75,9.25){u11}
\tkzDefPoint(5.25,9.25){u12}

\foreach \n in {u1,u5,u6,u10,u11}
  \node at (\n)[circle,fill,inner sep=2pt]{};
\foreach \n in {u2,u3,u4,u7,u8,u9,u12}
  \node at (\n)[circle,fill,inner sep=1.5pt]{};
\foreach \v/\w in {u1/u2, u1/u3, u2/u4, u3/u4, u2/u5, u3/u6, u4/u7, u5/u7, u6/u7, u6/u8, u7/u9, u7/u10, u8/u9, u8/u11, u10/u12, u11/u12, u9/u12}
    \draw[-] (\v) to (\w);
\end{tikzpicture}
}
\caption{Small instances inspired by the format instance from SteinLib, $T$ contains the larger nodes.\label{fig:small}}
\end{figure}

We illustrate and compare the above formulation on small artificial instances built upon the \emph{format} instance which includes 7 vertices and 9 edges (the instance is available at \url{http://steinlib.zib.de/format.php}). To get larger instances from the \emph{format} instance, we remove the central terminal and add layered copies of the instance. 
Figure~\ref{fig:small} depicts the original structure of the \emph{format} instance and that obtained by adding one copy. We denote as format$(\kappa)$ the instance with $\kappa$ copies of the original graph. 

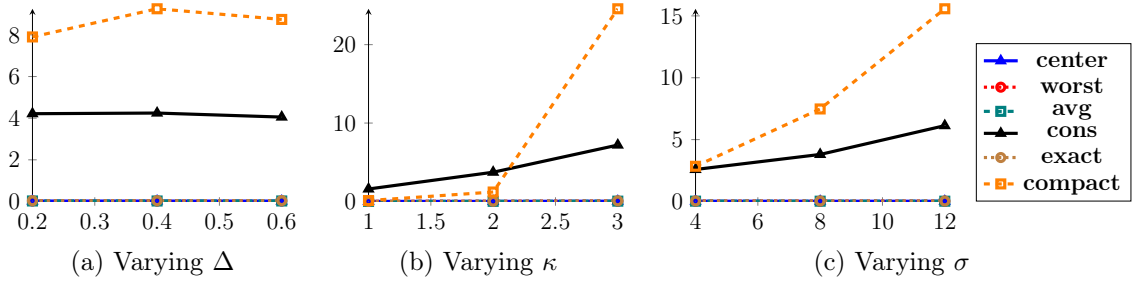
\begin{figure}[ht]
\centering
\subfloat[Varying $\Delta$]{
\begin{tikzpicture}[scale=0.75]
\begin{axis}[height=5cm,width=6cm,axis x line=bottom, axis y line = left,ymin = 0.0001,
	ylabel={},
	title={},
	%	ymode = log,
    log ticks with fixed point,
 	every axis plot/.append style={ultra thick}]
	
	\addplot [solid, blue, mark=triangle] table[
	x=x, y=center
	]{Steiner/small_times_Delta.txt};
	
	\addplot [dotted, red, mark=o, mark options=solid] table[
	x=x, y=worst
	]{Steiner/small_times_Delta.txt};
	
	\addplot [dashed, teal, mark=square, mark options=solid] table[
	x=x, y=avg
	]{Steiner/small_times_Delta.txt};
	
	\addplot [solid, black, mark=triangle] table[
	x=x, y=cons
	]{Steiner/small_times_Delta.txt};
	
	\addplot [dotted, brown, mark=o, mark options=solid] table[
	x=x, y=exact
	]{Steiner/small_times_Delta.txt};
	
	\addplot [dashed, orange, mark=square, mark options=solid] table[
	x=x, y=compact
	]{Steiner/small_times_Delta.txt};
\end{axis}
\end{tikzpicture}
}
\subfloat[Varying $\kappa$]{
\begin{tikzpicture}[scale=0.75]
\begin{axis}[height=5cm,width=6cm,axis x line=bottom, axis y line = left,ymin = 0.0001,
	ylabel={},
	title={},
	%	ymode = log,
    log ticks with fixed point,
 	every axis plot/.append style={ultra thick}]
	
	\addplot [solid, blue, mark=triangle] table[
	x=x, y=center
	]{Steiner/small_times_dim.txt};
	
	\addplot [dotted, red, mark=o, mark options=solid] table[
	x=x, y=worst
	]{Steiner/small_times_dim.txt};
	
	\addplot [dashed, teal, mark=square, mark options=solid] table[
	x=x, y=avg
	]{Steiner/small_times_dim.txt};
	
	\addplot [solid, black, mark=triangle] table[
	x=x, y=cons
	]{Steiner/small_times_dim.txt};
	
	\addplot [dotted, brown, mark=o, mark options=solid] table[
	x=x, y=exact
	]{Steiner/small_times_dim.txt};
	
	\addplot [dashed, orange, mark=square, mark options=solid] table[
	x=x, y=compact
	]{Steiner/small_times_dim.txt};
\end{axis}
\end{tikzpicture}
}
\subfloat[Varying $\sigma$]{
\begin{tikzpicture}[scale=0.75]
\begin{axis}[height=5cm,width=6cm,axis x line=bottom, axis y line = left,ymin = 0.0001,
	ylabel={},
	title={},
	%	ymode = log,
    log ticks with fixed point,
 	every axis plot/.append style={ultra thick},
 	legend entries={\hcenter, \worst, \avg, \cons, \exact, \compact},
	legend style={at={(axis cs:13,0.0001)}, anchor=south west}]
	
	\addplot [solid, blue, mark=triangle] table[
	x=x, y=center
	]{Steiner/small_times_nU.txt};
	
	\addplot [dotted, red, mark=o, mark options=solid] table[
	x=x, y=worst
	]{Steiner/small_times_nU.txt};
	
	\addplot [dashed, teal, mark=square, mark options=solid] table[
	x=x, y=avg
	]{Steiner/small_times_nU.txt};
	
	\addplot [solid, black, mark=triangle] table[
	x=x, y=cons
	]{Steiner/small_times_nU.txt};
	
	\addplot [dotted, brown, mark=o, mark options=solid] table[
	x=x, y=exact
	]{Steiner/small_times_nU.txt};
	
	\addplot [dashed, orange, mark=square, mark options=solid] table[
	x=x, y=compact
	]{Steiner/small_times_nU.txt};
\end{axis}
\end{tikzpicture}
}
\caption{STP: Average solution times in seconds on instances format($\kappa$) for each algorithm when varying one of the parameters.}
\label{fig:STP:times:small}
\end{figure}

The results presented on Figure~\ref{fig:STP:times:small} underline that \compact{} can hardly solve large instances, as the solution times increase significantly with $\kappa$. They also illustrate that \cons{} is much slower than \exact{} on these small artificial instances.

\section{Connection with the affine decision rules approximation from~\cite{zhen2021robust}}
\label{app:adr}

Notice first that, due to the convexity of the norm, the constraint 
$$
\forall \u\in\UU:  \sum_{\{i,j\}\in \EE} x_{ij} \|u_i-u_j\|_2 \leq \omega
$$
is equivalent to
$$
\forall \u\in\conv(\UU):  \sum_{\{i,j\}\in \EE} x_{ij} \|u_i-u_j\|_2 \leq \omega.
$$
Next, let us denote the unit ball of dimension $p$ by $\W$, as well as $\WW=\times_{e\in \EE}\W$. Let us also direct arbitrarily every edge in $\EE$, leading to the set of directed edges $\Edir$. Following the same idea as~\cite[Theorem~1]{zhen2021robust}, we obtain that the constraint 
$$
\forall \u\in\conv(\UU):  \sum_{(i,j)\in \Edir} x_{ij} \|u_i-u_j\|_2 \leq \omega
$$ 
is equivalent to
\begin{align}
&\forall \u\in\conv(\UU): \sum_{(i,j)\in \Edir}x_{ij} \max_{w_{ij}\in \W}w_{ij}^T(u_i-u_j) \leq \omega \\
\Leftrightarrow &\forall \w\in\WW, \u\in\conv(\UU): \sum_{(i,j)\in \Edir}x_{ij} w_{ij}^T(u_i-u_j) \leq \omega \\
\Leftrightarrow &\forall \w\in\WW : \max\set{\sum_{(i,j)\in \Edir}x_{ij} w_{ij}^T\left(\sum_{k=1}^\nU{i}\lambda_i^ku_i^k-\sum_{\ell=1}^\nU{j}\lambda_j^\ell u_j^\ell\right)}{\sum_{k=1}^\nU{i}\lambda_i^k=1, \forall i\in \VV,\, \lambda\geq0 } \leq \omega \\
\Leftrightarrow &\forall \w\in\WW : \min\set{\sum_{i\in \VV}\mu_i}{\mu_i \geq \left(\sum_{(i,j)\in \Edir} x_{ij}w_{ij}^T-\sum_{(j,i)\in \Edir} x_{ji}w_{ji}^T\right)u_i^k, \forall i\in \VV,k\in[\nU{i}]} \leq \omega. \label{eq:dualization}
\end{align}

Observe that the left-hand side of~\eqref{eq:dualization} can be interpreted as a two-stage robust optimization problem without first-stage variables, with $\mu$ playing the role of the second-stage variables, and with $w$ representing the uncertain parameters. This type of models being notoriously difficult to solve to optimality, we follow~\cite[Lemma~1]{zhen2021robust} and seek a heuristic solution by considering second-stage variables $\mu$ that can be expressed as affine decision rules
\begin{equation}
\label{eq:adr}
 \mu_i(\w)=\mu_i^0+\sum_{(i',j')\in \Edir} \mu_{i,i'j'}^T w_{i'j'},
\end{equation}
where $\mu_i^0\in\R$ and $\mu_{i,i'j'}\in\R^p$. Replacing~\eqref{eq:epigraphicL2} by~\eqref{eq:dualization} with $\mu$ substituted with the right-hand side of~\eqref{eq:adr}, we obtain
\begin{align*}
\min\quad & \omega \nonumber\\
\mbox{s.t.}\quad & \omega \geq \sum_{i\in \VV}\left(\mu_i^0+\sum_{(i',j')\in \Edir} \mu_{i,i'j'}^T w_{i'j'}\right), \quad \forall\w\in\WW\\
& \mu_i^0+\sum_{(i',j')\in \Edir} \mu_{i,i'j'}^T w_{i'j'} \geq \left(\sum_{(i,j)\in \Edir} x_{ij}w_{ij}^T-\sum_{(j,i)\in \Edir} x_{ji}w_{ji}^T\right)u_i^k, \quad\forall i\in \VV,k\in[\nU{i}], \w\in\WW\\
& x \in \X.\nonumber
\end{align*}
Dualizing the robust counstraints with respect to $w\in \WW$ yields
\begin{align}
\min\quad & \omega \label{eq:adr:first}\\
\mbox{s.t.}\quad & \omega \geq \sum_{i\in \VV}\mu_i^0+\sum_{(i,j)\in \Edir} \left\|\mu_{i,ij}+\mu_{j,ij}+\sum_{i'\neq i,j}\mu_{i',ij}\right\|_2\\
& \mu_i^0 \geq \sum_{(i,j)\in \Edir} \|x_{ij}u_i^k-\mu_{i,ij}\|_2+\sum_{(j,i)\in \Edir} \|x_{ji}u_i^k+\mu_{i,ji}\|_2 + \sum_{(i',j')\in \Edir:i\neq i',j'} \|\mu_{i,i'j'}\|_2, \quad\forall i\in \VV,k\in[\nU{i}]\label{eq:robustreformulationmu0}\\
& x \in \X.\label{eq:adr:last}
\end{align}
We discuss next how we can substantially reduce the number of affine multipliers in~\eqref{eq:adr}, and consequently, in problem~\eqref{eq:adr:first}--\eqref{eq:adr:last}. Let $(\tomega,\tx,\tmu^0,\tmu)$ denote an optimal solution to~\eqref{eq:adr:first}--\eqref{eq:adr:last}. Observe that the optimal solution cost is equal to $\tomega=\max_{u\in\UU}\tomega(u)$ where
\begin{align}
\tomega(u)&= \sum_{i\in \VV}\left(
\sum_{(i,j)\in \Edir} \|\tx_{ij}u_i-\tmu_{i,ij}\|_2+\sum_{(j,i)\in \Edir} \|\tx_{ji}u_i+\tmu_{i,ji}\|_2 + \sum_{(i',j')\in \Edir:i\neq i',j'} \|\tmu_{i,i'j'}\|_2
\right)\nonumber\\
&\qquad\qquad\qquad\qquad\qquad\qquad\qquad\qquad\qquad+
\sum_{(i,j)\in \Edir} \left\|\tmu_{i,ij}+\tmu_{j,ij}+\sum_{i'\neq i,j}\tmu_{i',ij}\right\|_2\\
&=\sum_{(i,j)\in \Edir}\left(
\|\tx_{ij}u_i-\tmu_{i,ij}\|_2
+\|\tx_{ij}u_j+\tmu_{j,ij}\|_2
+ \sum_{i'\neq i,j}\|\tmu_{i',ij}\|_2
+ \left\|\tmu_{i,ij}+\tmu_{j,ij}+\sum_{i'\neq i,j}\tmu_{i',ij}\right\|_2
\right)\label{eq:omegatilde}
\end{align}

We are going to define a sequence of two new solutions, $\tmu'$ and $\tmu''$, such that the corresponding values $\tomega'(u)$ and $\tomega''(u)$ are not greater than $\tomega(u)$ for each $u\in\UU$. First, we define $\tmu'$ by setting $\tmu'_{i',ij}=0$ for all $(i,j)\in \Edir$ such that $i'\notin\{i,j\}$ and $\tmu_{i',ij}'=\tmu_{i',ij}$ otherwise, and let $\tomega^{\prime}(u)$ be the right-hand side of~\eqref{eq:omegatilde} with $\mu$ replaced by $\mu'$.
Observe that
\begin{align*}
\sum_{i'\neq i,j}\|\tmu_{i',ij}'\|_2
+ \left\|\tmu_{i,ij}'+\tmu_{j,ij}'+\sum_{i'\neq i,j}\tmu_{i',ij}'\right\|_2
&=
\left\|\tmu_{i,ij}'+\tmu_{j,ij}'\right\|_2 \\
&= \left\|\tmu_{i,ij}+\tmu_{j,ij}\right\|_2 \\
&\leq \sum_{i'\neq i,j}\|\tmu_{i',ij}\|_2
+ \left\|\tmu_{i,ij}+\tmu_{j,ij}+\sum_{i'\neq i,j}\tmu_{i',ij}\right\|_2,
\end{align*}
which implies that 
\begin{equation}
\label{eq:primebetter}
\tomega^{\prime}(u)\leq\tomega(u), 
\end{equation}
for each $u\in\UU$. Second, we define another solution $\tmu''$ such that $\tmu''_{i,ij}=\tmu'_{i,ij}$ and $\tmu''_{j,ij}=-\tmu'_{i,ij}, \forall(i,j) \in\Edir$. 
Then, we denote as $\tomega^{\prime\prime}(u)$ the corresponding right-hand side of~\eqref{eq:omegatilde}. 
Observe that for each $u\in\UU$
\begin{align}
\tomega^{\prime\prime}(u)
&=\sum_{(i,j)\in \Edir}\left(
\|\tx_{ij}u_i-\tmu_{i,ij}''\|_2
+\|\tx_{ij}u_j+\tmu''_{j,ij}\|_2
+ \|\tmu_{i,ij}''+\tmu''_{j,ij}\|_2
\right)\\
&= \sum_{(i,j)\in \Edir}\left(
\|\tx_{ij}u_i-\tmu'_{i,ij}\|_2
+\|\tx_{ij}u_j-\tmu'_{i,ij}\|_2
\right)
\leq
\tomega^{\prime}(u).
\label{eq:primeprimebetter}
\end{align}

From~\eqref{eq:primebetter} and~\eqref{eq:primeprimebetter}, we see that we can set $\tmu_{i',ij}=0$ for all $(i,j)\in \Edir$ such that $i'\notin\{i,j\}$ and $\tmu_{i,ij}=-\tmu_{j,ij}, \forall(i,j)\in \Edir$ without deteriorating the quality of the solution returned by~\eqref{eq:adr:first}--\eqref{eq:adr:last}.
Thus, renaming the variables $\mu_{i,ij}$ as $\mu_{ij}$, formulation~\eqref{eq:adr:first}--\eqref{eq:adr:last} becomes
\begin{align*}
\min\quad & \omega \\
\mbox{s.t.}\quad & \omega \geq \sum_{i\in \VV}\mu_i^0\\
& \mu_i^0 \geq \sum_{(i,j)\in \Edir} \|x_{ij}u_i^k-\mu_{ij}\|_2+\sum_{(j,i)\in \Edir} \|x_{ji}u_i^k-\mu_{ji}\|_2 , \quad\forall i\in \VV,u\in\UU\\
& x \in \X,\label{eq:adr:last}
\end{align*}
and the equivalence with~\eqref{adr:first}--\eqref{adr:last} follows by removing the dummy variable $\omega$, introducing artificial variables to separate the norms into individual second-order cone constraints, and renaming $\mu_i^0$ as $d_i$.
\end{appendix}

%%

% References here (outcomment the appropriate case)

% CASE 1: BiBTeX used to constantly update the references
%   (while the paper is being written).
 % if more than one, comma separated

% CASE 2: BiBTeX used to generate mypaper.bbl (to be further fine tuned)
%\input{mypaper.bbl} % outcomment this line in Case 2

%If you don't use BiBTex, you can manually itemize references as shown below.

%%%%%%%%%%%%%%%%%
\end{document}